%% file: draft1.tex
\documentclass[letterpaper,11pt]{article}
\usepackage{amsmath, amsthm, amssymb}
\usepackage[utf8]{inputenc}
\usepackage[usenames, dvipsnames]{color}
\usepackage[normalem]{ulem} 
\usepackage{fullpage}
\usepackage{cite}
\usepackage[authoryear,round,longnamesfirst,sort]{natbib}
\usepackage[noend]{algorithmic}
\usepackage{tikz, pgfplots}
\usepackage{enumerate}
\usepackage{hyperref}
\usepackage{xspace,color}
\usepackage{bm} 
\usepackage{graphicx}
\usepackage{caption}
\usepackage{subcaption}
\usepackage{cleveref}
\usepackage{enumitem,linegoal}
\usepackage[linesnumbered,ruled,vlined]{algorithm2e}
\usepackage{comment}	
\usepackage{ctable}

\newtheorem{observation}{Observation}[section]

\usepackage{mathtools}
\usepackage[flushleft]{threeparttable}
\usepackage{verbatim}
\usepackage{setspace}
\usepackage{bbm}
\usepackage{thm-restate}
\usepackage{footnote}
\usepackage{tcolorbox}
\usepackage{fdsymbol}
\usepackage{nicefrac}
\usepackage{array}
\usepackage{booktabs}
\usepackage{colortbl}
\usepackage{multirow}
\usepackage{listings}
\usepackage{apxproof}
\usepackage{tikz-3dplot}
\usetikzlibrary{arrows.meta, positioning, decorations.markings,shadows, shapes.symbols}
\colorlet{transblue}{blue!40!white}
\colorlet{transred}{red!40!white}
\colorlet{transgreen}{green!40!white}

\usepackage[framemethod=TikZ]{mdframed}
\newmdenv[
backgroundcolor=gray!5,
linecolor=black,
linewidth=0.7pt,
roundcorner=6pt,
innertopmargin=12pt,
innerbottommargin=12pt,
innerleftmargin=10pt,
innerrightmargin=10pt,
skipabove=12pt,
skipbelow=12pt
]{fancybox}

\makesavenoteenv{tabular}
\makesavenoteenv{table}
\crefname{enumi}{Condition}{Conditions}

\usepackage[margin=1in]{geometry}
\definecolor{darkgreen}{rgb}{0.0, 0.39, 0.0}

\hypersetup{
	colorlinks   = true, 
	urlcolor     = blue, 
	linkcolor    = blue, 
	citecolor   = darkgreen 
}

\definecolor{crimsonglory}{rgb}{0,0,0}

\newtheorem{example}{Example}

\newtheorem{theorem}{Theorem}[section]

\newtheorem{corollary}[theorem]{Corollary}

\newtheorem{definition}[theorem]{Definition}

\makeatletter
\def\GrabProofArgument[#1]{ #1: \egroup\ignorespaces}
\def\proof{\noindent\textbf\bgroup Proof%
	\@ifnextchar[{\GrabProofArgument}{. \egroup\ignorespaces}}

\makeatother

\usetikzlibrary{shapes.geometric, arrows.meta, positioning,shadows,calc,decorations.markings,patterns,graphs,mindmap, decorations.pathmorphing,fit,shapes.misc,backgrounds}

\tikzstyle{startstop} = [rectangle, rounded corners, minimum width=3cm, minimum height=1cm, text centered, draw=black, fill=red!30]
\tikzstyle{process} = [rectangle, minimum width=3cm, minimum height=1cm, text centered, draw=black, fill=orange!30]
\tikzstyle{decision} = [diamond, minimum width=3cm, minimum height=1cm, text centered, draw=black, fill=green!30]
\tikzstyle{arrow} = [thick,->,>=stealth]
\tikzset{
agentj/.style={circle, draw, minimum size=7mm, font=\small, inner sep=1pt, fill=agentj}, 
agenti/.style={circle, draw, minimum size=7mm, font=\small, inner sep=1pt, fill=agenti}, 
bundle/.style={rectangle, draw=orange!70!black, rounded corners,
minimum width=0.7cm, minimum height=0.5cm, fill=orange!30, opacity=1, font=\scriptsize, align=center},
heavy/.style={-{Stealth[length=2mm]}, thick, red, bend left=15},
heavyrev/.style={-{Stealth[length=2mm]}, thick, red, bend right=15}, 
light/.style={-{Stealth[length=2mm]}, thick, green!70!black, bend left=15},
rotation/.style={-{Stealth[length=2mm]}, thick, magenta},
		choice/.style={->, dashed},
subpool/.style={rectangle, draw=blue!70!black, thick,
	rounded corners, minimum width=1cm, minimum height=0.6cm,
	fill=white, opacity=1, font=\scriptsize, align=center},
smallset/.style={rectangle, draw=green!70!black, thick,
	rounded corners, minimum width=0.8cm, minimum height=0.4cm,
	fill=green!20, opacity=1},
poolimg/.style={inner sep=0pt, anchor=center, opacity=0.8, minimum width=2.5cm, minimum height=1.0cm},
rotind/.style={circle, draw=white, thick, minimum size=8mm},
ret/.style={    decorate,
	decoration={
		snake,               
		amplitude=1.5pt,     
		segment length=6pt,  
		post length=7pt,     
		pre length=3pt       
	},
	thick,                 
	-{Stealth[length=2mm]},
	blue  },
	good/.style={
		rectangle,
		draw=black!80,
		fill=white,
		rounded corners=3pt,
		font=\small\bfseries,
		inner sep=2pt,
		text=black
	}
}
\definecolor{skyblue}{RGB}{135,206,235}  
\definecolor{coral}{RGB}{255,127,80}   
\definecolor{limegreen}{RGB}{50,205,50}    
		\definecolor{agent1}{RGB}{173,216,230}
\definecolor{agent2}{RGB}{255,182,193}
\definecolor{agent3}{RGB}{152,251,152}

\input{src/macro.tex}

\newcounter{proccnt}

\newcommand{\konote}[1]{}

\usepackage[normalem]{ulem} 

\title{Tight Bounds On the Distortion of Randomized and Deterministic Distributed Voting}

\author{
	MohammadAli Abam\thanks{Sharif University of Technology, Tehran, Iran}
	\and
	Davoud Kareshki\footnotemark[1] \and
	Marzieh Nilipour\footnotemark[1] \and
	MohammadHossein Paydar\footnotemark[1] \and
	Masoud Seddighin\thanks{Tehran Institute for Advanced Studies (TeIAS), Khatam University, Tehran, Iran}
}

			\definecolor{darkgreen}{rgb}{0.0, 0.5, 0.0}
		\definecolor{agentk}{RGB}{152,251,152}  
\definecolor{agenti}{RGB}{173,216,230}  
\definecolor{agentj}{RGB}{255,182,193}  
\begin{document}
	\newcommand{\ignore}[1]{}
	\renewcommand{\theenumi}{\roman{enumi}.}
	\renewcommand{\labelenumi}{\theenumi}
	\sloppy
	\date{} 
	\newenvironment{subproof}[1][\proofname]{
		\renewcommand{\Box}{ \blacksquare}%
		\begin{proof}[#1]%
		}{%
		\end{proof}%
	}
	\maketitle
	
	\thispagestyle{empty}
	\allowdisplaybreaks
	
	\vspace{0.5cm}
	\begin{abstract}
		\input{src/abstract.tex}

	\end{abstract}
	\section{Introduction}	
	\input{src/intro.tex}

		\subsection{Further Related Work}
		\input{src/related.tex}

		\subsection{Our Contributions}\label{sec:results}
		\input{src/contributions.tex}

	\section{Basic Notations}	
	\input{src/preliminaries.tex}
	\section{Distortion Bounds for \randdet}\label{sec:randdet}
	\input{src/randdet.tex}

		\subsection{Upper Bounds}

\input{src/randdet_upper.tex}
		\subsection{Lower Bounds}
		\input{src/randdet_lower.tex}
	\section{Distortion Bounds for \randrand}\label{sec:randrand}
	\input{src/randrand.tex}

		\subsection{Upper Bounds}
		\input{src/randrand_upper.tex}
		\subsection{Lower Bounds}
		\input{src/randrand_lower.tex}

	\section{Resolving the Distortion Bounds for \detdet}\label{sec:detdet}
	\input{src/detdet.tex}

		\subsection{Upper Bounds}
		\input{src/detdet_upper.tex} 
		\subsection{Lower Bounds}
		\input{src/detdet_lower.tex} 
	\section{An Extension of Lower Bounds for \randrand\ and \randdet}\label{sec:euclidean}
	\input{src/euclidean.tex}
	\section{Discussion and Open Problems}
	\input{src/discussion}

	\newpage
	\bibliographystyle{apalike}
	\bibliography{ref}
	\appendix
\end{document}

%% file: src/macro.tex
\newcommand{\instance}{\mathcal{I}} 
\newcommand{\voters}{\mathcal{V}} 
\newcommand{\voter}{v}
\newcommand{\candidates}{\mathcal{C}} 
\newcommand{\candidate}{c}
\newcommand{\groups}{\mathcal{G}} 
\newcommand{\group}{g}
\newcommand{\profile}{\pi} 
\newcommand{\mech}{{\Psi}} 
\newcommand{\dis}[2]{{\small\textsf{d}(#1,#2)}} 
\newcommand{\distortion}{\mathsf{D}} 
\newcommand{\size}{{n}} 
\newcommand{\w}{\textsf{w}} 
\newcommand{\rep}{\textsf{w}} 
\newcommand{\opt}{\textsf{o}} 
\newcommand{\randdet}{{\textsf{rand-det}}}
\newcommand{\randrand}{{\textsf{rand-rand}}}
\newcommand{\detdet}{{\textsf{det-det}}}
\newcommand{\detrand}{\textsf{det-rand}}
\newcommand{\avgavg}{{\small\textsf{avg-avg}}}
\newcommand{\avgmax}{{\small\textsf{avg-max}}}
\newcommand{\maxavg}{{\small\textsf{max-avg}}}
\newcommand{\maxmax}{{\small\textsf{max-max}}}
\newcommand{\maxx}{{\small\textsf{max}}}
\newcommand{\avgg}{{\small\textsf{avg}}}
\newcommand{\topp}[1]{\textsf{top}(#1)} 
\newcommand{\unusedvar}{l}
\newcommand{\cost}{\textsf{cost}} 
\newcommand{\expected}[1]{\mathrm{E}\!\left[#1\right]}

\newcommand{\vstar}[2]{\voter^{*}(#1,#2)} 
\newcommand{\vstarstar}[1]{\voter^{**}(#1)} 
\newcommand{\gstar}{g^{*}} 
\newcommand{\fin}{\mathsf{f}_{in}} 
\newcommand{\fov}{\mathsf{f}_{ov}} 
\newcommand{\falpha}{\mathsf{f}_{\alpha}}
\newcommand{\fur}{\mathsf{f}_{ur}} 
\newcommand{\fun}{\mathsf{f}_{par}} 
\newcommand{\fpmpar}{\mathsf{f}_{pm-par}} 
\newcommand{\fpm}{\mathsf{f}_{pm}} 
\newcommand{\frd}{\mathsf{f}_{rd}} 
\newcommand{\mad}{\mathsf{m}_{ad}} 
\newcommand{\malphabeta}{\mathsf{m}_{\alpha\beta}} 
\newcommand{\fbeta}{\mathsf{f}_{\beta}} 
\newcommand{\tour}[3]{\mathcal{T}(#1,#2,#3)} 
\newcommand{\movetofirst}[2]{#1 {\small\uparrow} #2} 

%% file: src/abstract.tex
%
%
%

We study metric distortion in distributed voting, where $n$ voters are partitioned into $k$ groups, each selecting a local representative, and a final winner is chosen from these representatives (or from the entire set of candidates). This setting models systems like U.S. presidential elections, where state-level decisions determine the national outcome. We focus on four cost objectives from \citep{anshelevich2022distortion}: $\avgavg$, $\avgmax$, $\maxavg$, and $\maxmax$. We present improved distortion bounds for both deterministic and randomized mechanisms, offering a near-complete characterization of distortion in this model.

For deterministic mechanisms, we reduce the upper bound for $\avgmax$ from $11$ to $7$, establish a tight lower bound of $5$ for $\maxavg$ (improving on $2+\sqrt{5}$), and tighten the upper bound for $\maxmax$ from $5$ to $3$.

For randomized mechanisms, we consider two settings: (i) only the second stage is randomized, and (ii) both stages may be randomized. In case (i), we prove tight bounds: $5\!-\!2/k$ for $\avgavg$, $3$ for $\avgmax$ and $\maxmax$, and $5$ for $\maxavg$. In case (ii), we show tight bounds of $3$ for $\maxavg$ and $\maxmax$, and nearly tight bounds for $\avgavg$ and $\avgmax$ within $[3\!-\!2/n,\ 3\!-\!2/(kn^*)]$ and $[3\!-\!2/n,\ 3]$, respectively, where $n^*$ denotes the largest group size. 

%% file: src/intro.tex
\label{sec:intro}

In social choice theory, a voting rule is a function that takes agents' preferences over alternatives and selects one as the final outcome. Preferences are usually represented as ranked lists, and the goal is to design a voting rule that best reflects these preferences.

How can we evaluate whether a voting rule is appropriate? There are both \emph{axiomatic} and \emph{quantitative} benchmarks for assessing outcomes \citep{arrow2010handbook,brandt2016handbook,procaccia2006distortion}. In this paper, we focus on one of the most prominent quantitative measures: \emph{distortion}. The idea is simple: each agent has hidden numerical values—either \emph{costs} or \emph{utilities}—for the alternatives, and their ordinal rankings reflect these values. Suppose our ultimate goal is to optimize an objective function $\phi$, such as social cost, maximum cost, or total utility, based on these hidden values. However, since the voting rule only has access to the agents' ordinal preferences—not the numerical values—it may select a suboptimal outcome with respect to $\phi$.
The distortion of a voting rule captures how far its chosen outcome can be from the optimal one in the worst case. It is defined as the ratio between the value of $\phi$ for the selected outcome and the value of $\phi$ for the optimal alternative in the worst case.

Since its introduction by \citet{procaccia2006distortion}, distortion has been a consistent focus of investigation—not only in voting, but also in related social choice problems such as facility location \citep{feldman2016voting,chan2021mechanism,anshelevich2021ordinal,kanellopoulos2023truthful} and matching \citep{anshelevich2021ordinal,amanatidis2022few,anari2023distortion,latifian2024distortion}. Still, the core of the literature lies in voting, with particularly rich results when costs of agents form a metric space \citep{anari2023distortion,anshelevich2018approximating,munagala2019improved,gkatzelis2020resolving,anshelevich2017randomized,charikar2022metric,caragiannis2022metric,anagnostides2022metric,kizilkaya2023generalized,bagheridelouee2024metric,ghodsi2019distortion}. For a comprehensive overview of distortion in voting, we refer to the survey by \citet{anshelevich2021distortion}.

In this paper, we study the distortion in the metric setting when the voting process is distributed. Unlike centralized voting, in many large-scale scenarios, outcomes emerge via a two-stage manner: decisions are made locally within separate groups of agents, the local outcomes are then aggregated into a final outcome. A notable example is the U.S. presidential election, where each state selects a winner, and the national outcome is determined by a weighted aggregation of the state-level results. 
More formally, a distributed voting mechanism is a pair $(\fin, \fov)$, where
\begin{enumerate}
	\item $\fin$ is an \emph{in-group} voting rule that selects a local winner for each group based solely on the preferences of agents within that group.
	\item Assuming $R$ is the set of local winners, $\fov$ is an \emph{over-group} voting rule that selects the final winner based on the preferences of $R$ over all alternatives \citep{anshelevich2022distortion} or local winners \citep{filos2024revisiting}.
\end{enumerate}

The study of distortion in distributed voting was pioneered by \citet{filos2020distortion}, who extended the notion of distortion to the utility-based distributed scenario. Later, \citet{anshelevich2022distortion} investigated distributed voting in the metric cost setting.
In the context of distributed voting, since decisions occur in two stages, it is natural to define separate cost objectives for each level. Building on this, \citet{anshelevich2022distortion} introduced four objectives combining \emph{average} and \emph{maximum} costs within and across the groups: $\avgavg$, $\maxavg$, $\avgmax$, and $\maxmax$. 
In the deterministic setting, they proved constant upper and lower bounds for all objectives in both general and line metric spaces, summarized in \Cref{tab:detresults}.
Later, \citet{voudouris2023tight} focused on the line metric and proposed two simple mechanisms for the \avgmax\, and \maxavg\, objectives. These mechanisms achieve an upper bound of $2+\sqrt{5}$, closing the corresponding gap derived by \citet{anshelevich2022distortion}.

\begin{table}[t]
	\centering
	    \renewcommand{\arraystretch}{1.3}
	\begin{tabular}{|c|c|c|}
		\hline
		\rowcolor{gray!20}
		\small{Cost objective} & \small{General metric} & \small{Line metric} \\
		\hline
		\hline
		\avgavg & $[7, 11]$ \citep{anshelevich2022distortion} & $7$ \citep{anshelevich2022distortion} \\
		\hline
		\avgmax & $[2 + \sqrt{5}, 11]$ \citep{anshelevich2022distortion} & $2 + \sqrt{5}$ \citep{voudouris2023tight} \\
		\hline
		\maxavg & $[2 + \sqrt{5}, 5]$ \citep{anshelevich2022distortion} & $2 + \sqrt{5}$ \citep{voudouris2023tight} \\
		\hline
		\maxmax & $[3, 5]$ \citep{anshelevich2022distortion} & $3$ \citep{anshelevich2022distortion} \\
		\hline
	\end{tabular}
	\vspace{0.2cm}
	\caption{A summary of deterministic results for distributed mechanisms with respect to different cost objectives. Each single value indicates a tight bound.}
	\label{tab:detresults}
\end{table}

As shown in \Cref{tab:detresults}, distributed voting on the line metric is well-understood, with tight distortion bounds already achieved. We therefore turn to general metric spaces and explore whether randomization can also improve distortion in the distributed setting. This work presents the first formal investigation into randomized distributed mechanisms within the metric setting.

In this paper, we make significant progress on the distortion of distributed voting mechanisms in two main directions. First, we improve the existing distortion bounds of \emph{deterministic} mechanisms with respect to the $\avgmax$, $\maxavg$, and $\maxmax$ objectives.
Second, we explore rules that incorporate \emph{randomized} mechanisms—either in the second stage only, or in both stages—referred to as $\randdet$ and $\randrand$, respectively. The output of a randomized mechanism is a probability distribution over the alternatives, rather than a single winner. 
For both the $\randdet$ and $\randrand$ mechanisms, we prove tight bounds for almost all of the objectives. 

%% file: src/related.tex

The most relevant studies to our work \citep{anshelevich2022distortion,voudouris2023tight} are discussed in \Cref{sec:intro}. Here, we briefly review other related studies. Since Procaccia and Rosenschein’s seminal work (\citeyear{procaccia2006distortion}), research on distortion in social choice problems has expanded, covering utilitarian settings \citep{boutilier2015optimal,caragiannis2011voting,ebadian2023explainable,ebadian2024optimized,bedaywi2025distortion}, metric settings \citep{anshelevich2018approximating,gkatzelis2020resolving,jaworski2020evaluating,kizilkaya2022plurality,charikar2022metric,caragiannis2022metric,charikar2024breaking,bagheridelouee2024metric}, and combined approaches \citep{gkatzelis2023best}.

\paragraph{Deterministic mechanisms.} 
\citet{anshelevich2018approximating} pioneered the study of distortion for the metric framework. Using a simple example, they show that the distortion of any deterministic voting rule is at least $3$. \citet{gkatzelis2020resolving} proposed an elegant and intricate voting rule, \emph{Plurality Matching}, which achieves a tight distortion of $3$. Next, \citet{kizilkaya2022plurality} attained the same upper bound with a simpler voting rule, \emph{Plurality Veto}. 

\citet{filos2020distortion} pioneered distortion analysis in distributed voting under the utilitarian framework. 
Their work extended to other social choice problems, including facility location \citep{filos2024distortion}, aiming to select a single location from a set of alternatives. 
More recently, \citet{voudouris2025metric} investigated the distributed distortion in obnoxious voting, where alternatives are undesirable. 

\paragraph{Randomized mechanisms.}
Unlike deterministic voting rules, randomized rules can achieve distortion below~3. 
\citet{anshelevich2017randomized} prove that the metric distortion of \emph{Random Dictatorship} is at most $3 - 2/n$, whit $n$ agents, and establish a lower bound of 2 for any randomized voting rule. \citet{kempe2020communication} improves the upper bound for \emph{Random Dictatorship} to $3 - 2/m$, where $m$ is the number of candidates.
\citet{charikar2022metric} further raise the lower bound for any randomized rule to $2.112$. 
Recently,~\citet{charikar2024breaking} reduce the upper bound to $2.753$.

In the context of distributed voting, \citet{filos2024revisiting} investigated randomized mechanisms under the utilitarian framework, establishing distortion bounds in various cases.
In ordinal setting, they proved a distortion of $\Theta(km^{2})$ for \emph{randomized-of-deterministic} mechanisms, where $k$ denotes the number of groups and $m$ the number of candidates. For \emph{randomized-of-randomized} mechanisms, they showed that the distortion is bounded between $\Omega(\sqrt{m})$ and $O(\sqrt{m\log m})$. They also introduced strategyproof mechanisms that achieve low distortion.

%% file: src/contributions.tex

Our results provide improved upper and lower bounds on the distortion of distributed mechanisms across various combinations of deterministic and randomized voting rules and cost objectives. As summarized in \Cref{tab:ourresults}, most of our bounds are tight—despite the fact that our proposed mechanisms are simple. 
In addition to general metric spaces, we also analyze the Euclidean setting and derive corresponding bounds under this restriction.

\paragraph{Randomized Distributed Mechanisms.}

Previous work on metric distortion in the distributed setting has focused exclusively on deterministic voting rules \citep{anshelevich2022distortion,voudouris2023tight,amanatidis2025metric,voudouris2025metric}. In this paper, we take a significant step toward understanding  randomized mechanisms in distributed voting. 
We study two natural classes of randomized mechanisms---\randdet\ and \randrand---within general metric spaces, and analyze their performance with respect to all the four objectives. See \Cref{tab:ourresults} for an overview of our results.

\textbf{\randdet\ mechanisms}, defined as pairs $(\fin, \fov)$, where $\fin$ is a deterministic voting rule and $\fov$ is a randomized one.
We derive several tight distortion bounds with respect to the all objectives in \Cref{sec:randdet}.

\begin{itemize}
	\item \textbf{\maxmax, \avgmax:} For both objectives, we derive a tight distortion bound of $3$. The lower bound is established through a basic example within a single group on a line metric, simplifying the \maxmax\, and \avgmax\, objectives to \maxx.
	The upper bound is proven by a distributed mechanism that first selects a representative for each group with the \emph{Plurality Matching} rule and then chooses the final winner uniformly at random.
	
	\item \textbf{\maxavg:} We establish a tight distortion bound of $5$. The lower bound is proven using a line metric and a novel tool we introduce, called the \emph{Bias Tournament}, which may be of independent interest. For the upper bound, we show that applying a deterministic in-group rule with a distortion at most $\alpha \ge 3$, followed by the \emph{Random Dictatorship} rule\footnote{Refer to \Cref{sec:pre} for the formal definition.}, achieves an overall distortion of at most $\alpha+2$. Since the best achievable value of $\alpha$ is 3 (via the \emph{Plurality Matching} rule), this yields a matching upper bound of 5.
	
	\item \textbf{\avgavg:} We prove a tight distortion bound of $5 - \nicefrac{2}{k}$. Obtaining this bound for the \avgavg\, objective is the most challenging aspect of the \randdet\, mechanisms. The lower bound construction, though similar to that of the \maxavg\, objective, requires a more delicate analysis to extract the $\nicefrac{2}{k}$ improvement. Once again, we employ the Bias Tournament and model the metric space via shortest-path distances in a graph.
\end{itemize}

In \Cref{sec:randrand}, we analyze \textbf{\randrand\ mechanisms}, defined as pairs $(\fin, \fov)$ comprising of two randomized voting rules, and derive tight or near-tight distortion bounds.
All the upper bounds are obtained via a distributed mechanism that initially applies the \emph{Random Dictatorship} rule within each group and then randomly selects the final winner from the chosen representatives with uniform probability.

\begin{itemize}
	\item \textbf{\maxmax, \maxavg:} For both objectives, we establish a tight distortion bound of $3$. We construct a shared example consisting of $k$ single-voter groups to establish the lower bound, even when the metric space is a line. In this scenario, the \maxmax\, and \maxavg\, objectives both simplify to the \maxx\ objective.
	 
	\item \textbf{\avgmax:} We establish a lower bound of $3-\tfrac{2}{n}$, which nearly matches our upper bound of $3$. The lower bound is proven with an instance where the number of candidates and voters are equal ($n=m$) and there is only a single group ($k=1$).
	Additionally, we conclude a lower bound for the $\maxx$ objective in the centralized setting: We show that any randomized voting rule must have a distortion of at least $3 - \varepsilon$ for any constant $\varepsilon > 0$. This is particularly interesting since even deterministic rules are known to have an upper bound of 3 for the \maxx\ objective \citep{gkatzelis2020resolving}.
	
	\item \textbf{\avgavg:} We establish a nearly tight distortion bound slightly below $3$. For an instance with $k$ single-voter groups on a tree graph, we prove lower bound of $3-\tfrac{2}{n}$. We further derive an upper bound of $3-\tfrac{2}{kn^*}$, where $n^*$ denotes the largest group size. When all groups are of equal size, it yields matching upper and lower bounds. Notably, deriving these bounds is the most challenging aspect of analyzing \randrand, mechanisms.
\end{itemize}

\begin{table}[t]
	\centering
	\begin{tabular}{|l | c | c c |}
		\hline
		\rowcolor{gray!20}
		& Objective & \multicolumn{2}{c|}{Distortion} \\
		\rowcolor{gray!20}
		&  & \scriptsize lower bound & \scriptsize upper bound \\
		\hline
		\hline
		\multirow{4}{*}{\rotatebox{90}{\detdet*}}
		& \avgavg & \textcolor{gray}{$7$ \citep{anshelevich2022distortion}} & \textcolor{gray}{$11$} \citep{anshelevich2022distortion}  \\
		& \avgmax & \textcolor{gray}{$2 + \sqrt{5}$ \citep{anshelevich2022distortion}} & $7$ \scriptsize(\Cref{cor:detdet_avgmax})  \\
		& \maxavg & $5$~\scriptsize(\Cref{th:detdet-maxavg-lower}) & $\textcolor{gray}{5}$ \citep{anshelevich2022distortion} \\
		& \maxmax & $\textcolor{gray}{3}$ \citep{anshelevich2022distortion} & $3$~\scriptsize(\Cref{th:detdet_maxmax}) \\
		\hline
		\multirow{4}{*}{\rotatebox{90}{\randdet}}
		& \avgavg & $5-\frac{2}{k}$~\scriptsize(\Cref{th:randdet-avgavg-lower}) & $5-\frac{2}{k}$~\scriptsize(\Cref{cor:randdet-avgavg-upper})\\
		& \avgmax & $3$~\scriptsize(\Cref{th:randdet-Xmax-lower}) & $3$~\scriptsize(\Cref{th:Xmaxuni}) \\
		& \maxavg & $5$~\scriptsize(\Cref{th:randdet-maxavg-lower}) & $5$~\scriptsize(\Cref{cor:randdet-maxavg-upper}) \\
		& \maxmax & $3$~\scriptsize(\Cref{th:randdet-Xmax-lower}) & $3$~\scriptsize(\Cref{th:Xmaxuni}) \\
		\hline
		\multirow{4}{*}{\rotatebox{90}{\randrand}} 
		& \avgavg & $3-\frac{2}{n}$~\scriptsize(\Cref{th:randrand-avgavg-lower}) & $3-\frac{2}{kn^*}$~\scriptsize(\Cref{th:randrand_avgavg})\\
		& \avgmax & $3-\frac{2}{n}$~\scriptsize(\Cref{th:randrand-avgmax-lower}) & $3$~\scriptsize(\Cref{th:randrand_avgmax}) \\
		& \maxavg & $3$~\scriptsize(\Cref{th:randrand-maxX-lower}) & $3$~\scriptsize(\Cref{th:randrand_maxavg}) \\
		& \maxmax & $3$~\scriptsize(\Cref{th:randrand-maxX-lower}) & $3$~\scriptsize(\Cref{th:randrand_maxmax}) \\
		\hline
	\end{tabular}
	\vspace{0.2cm}
	\caption{An overview of our results for various cost objectives in general metric spaces, with gray-text results indicating those derived from prior work. $n^*$ denotes the size of the largest group. Thus, the bound of $3\!-\!2/n$ for the \avgavg\, objective in \randrand\ is tight when all group sizes are equal. *Note: For the \detdet\ mechanisms, we follow the setting of \citep{anshelevich2022distortion}, where the over-group rule is applied to all candidates.}
	\label{tab:ourresults}
\end{table}

\paragraph{Deterministic Distributed Mechanisms.}
We consider \textbf{\detdet\ mechanisms}, defined as pairs $(\fin, \fov)$ comprising of two independently deterministic voting rules, in \Cref{sec:detdet}. 
We resolve the previously known gaps for the \maxavg\ and \maxmax\ objectives and provide an enhanced upper bound for the \avgmax\ objective.
In this section, we adopt a setting akin to \citep{anshelevich2022distortion}, where $\fov$ selects a winner from the set of \emph{all candidates}, not solely those chosen in the first stage.

\begin{itemize}
	\item \textbf{\avgmax.} 
	We improve the upper bound from \textbf{11} to \textbf{7}. \citet{anshelevich2022distortion} show that combining the in-group and over-group voting rules with distortions $\alpha$ and $\beta$, respectively, yields an overall distortion of $\alpha + \beta + \alpha\beta$. With their best known values ($\alpha = 3$, $\beta = 2$), this gives $11$. We prove that if the in-group rule, $\fin$, merely satisfies the property of pareto efficiency, then the overall distortion is at most $2\beta + 3$, which is \emph{independent} of $\alpha$. This results in a tighter upper bound of $7$, and shows the dominant role of the over-group rule in this setting.
	
	\item \textbf{\maxavg.} 
	We improve the lower bound from $2 + \sqrt{5}$ to \textbf{5}. Interestingly, our lower-bound instance is based on a metric constructed via shortest-path in a graph, rather than a line or Euclidean. This confirms that the upper bound from \citep{anshelevich2022distortion} is indeed tight.	
	
	\item \textbf{\maxmax.} 
	We improve the upper bound on distortion from $\textbf{5}$ to $\textbf{3}$ for general metric spaces. Although a bound of $3$ was previously known for the \emph{line metric}, the general case remained open. We show that a distributed mechanism same as the \emph{Arbitrary Dictator}, proposed by \citet{anshelevich2022distortion}, actually achieves distortion $3$ for any metric space.
\end{itemize}

\paragraph{Bias Tournament.}
The Bias Tournament is a directed graph with one node per candidate. For any pair of candidates \( c_1 \) and \( c_2 \), we add a directed edge from \( c_1 \) to \( c_2 \) if, in a group containing only two voters with preferences \( (c_1, c_2, \ldots) \) and \( (c_2, c_1, \ldots) \)—where the remaining candidates are ordered identically according to a fixed permutation \( \sigma \) over all candidates—the in-group rule \( \fin \) deterministically selects \( c_1 \) as the winner. This construction captures the bias in \( \fin \)'s tie-breaking behavior across candidate pairs.

To analyze the implications of these biases, we use a well-known property of tournaments: any tournament on \( m \) nodes contains at least one node with in-degree at least \( \left\lceil \nicefrac{(m-1)}{2} \right\rceil \). This fact allows us to identify a \emph{losing} candidate—one frequently defeated in pairwise comparisons—and use it as the basis for constructing an instance with high distortion, which helps us to establish the lower bounds for the \randdet\ and \detdet\ mechanisms.

\paragraph{Centralized  Setting.}
Interestingly, one of our lower bound constructions, originally developed for distributed mechanisms also applies to randomized centralized setting. In particular, we analyze the $\max$ cost objective in this setting. Previously, \citet{gkatzelis2020resolving} established an upper bound of 3 for $\max$ using \emph{Plurality Matching}, as defined in \Cref{sec:pre}. In our work, we revisit the $\avgmax$\ objective for the $\randrand$\ mechanisms and prove a distortion lower bound. We then reuse the same construction to show that the distortion of $\max$\ under centralized voting is at least \( 3 - \varepsilon \) for any constant \( \varepsilon > 0 \), thus matching the known upper bound up to an arbitrarily small gap.

\paragraph{Euclidean Space.} 
Our lower-bound constructions for the $\avgavg$ objective rely on general metric spaces that do not reduce to simpler structures such as the line or Euclidean space. This raises the natural question of whether similarly strong bounds can be achieved in more structured settings. In \Cref{sec:euclidean}, we address this by constructing novel instances within a Euclidean hyper-simplex, obtaining lower bounds of \( \sqrt{5} - \varepsilon \) for the $\randrand$ mechanisms and \( 2 + \sqrt{5} - \varepsilon \) for the $\randdet$ ones, where \( \varepsilon \) is arbitrarily small. 
Specifically, for $\randrand$, we fix an integer $\unusedvar$ and construct an instance with $\unusedvar+2$ candidates and $k=\unusedvar+1$ single-voter groups embedded in $\mathbb{R}^{\unusedvar+1}$. Similarly, for $\randdet$, we construct an instance with $2m$ candidates and $m$ groups of two voters each, embedded in $\mathbb{R}^{m+1}$ (where $m$ is fixed).

%% file: src/preliminaries.tex
\label{sec:pre}

An instance of the \emph{distributed voting problem} is a tuple $(\voters, \candidates, \groups, \profile, \textsf{d})$. Here, $\voters$ is a set of $n$ voters, and $\candidates$ is a set of $m$ candidates. We use $v_i$ to denote the $i$-th voter and symbols $a$, $b$, $c$, $\opt$, and $\rep$ to refer to candidates. In particular, $\w$ typically denotes the winner and $\opt$ refers to the optimal candidate.

$\groups$ is a partition of voters into $k$ groups, such that each voter belongs to exactly one group. For each group $g \in \groups$, $\size_g$ is the number of voters in $g$. We denote the optimal candidate for group \( g \) by \( \opt_g \). We say an instance is \emph{symmetric}, if all groups have equal sizes. 

Each voter $v_i$ has a strict ranking $\pi_i$ over the set of candidates, representing their ordinal preferences. A preference profile is the collection of preferences from all voters, denoted by $\profile = (\pi_1, \ldots, \pi_n)$. These preferences arise from a cost function based on the underlying metric space~$\textsf{d}$. 
For each voter $\voter$ and candidate $\candidate$, $\dis{\voter}{\candidate}$ denotes the cost of candidate $\candidate$ for voter $\voter$. The voters rank all candidates in increasing order of cost, which means they prefer those who are closer.
The distance function $\small\textsf{d}$ satisfies the standard metric properties of {\emph{non-negativity}}, {\emph{symmetry}}, and the \emph{triangle inequality}. 

We denote the top-ranked candidate of each voter \( \voter \) as \( \topp{\voter} \).
For a candidate \( c \) and a group \( g \in \groups \), \( \vstar{c}{g} = \arg\max_{\voter \in g} \dis{\voter}{c} \) denotes the farthest voter from \( c \) within group \( g \). Similarly, \( \vstarstar{c} = \arg\max_{\voter \in \voters} \dis{\voter}{c} \) represents the farthest voter from \( c \) across all voters. 
We also use \( \cost_g(c) \) to denote the cost of candidate \( c \) restricted to group \( g \), and define \( \gstar(c) = \arg\max_{g \in \groups} (\nicefrac{1}{\size_g}) \sum_{\voter \in g} \dis{\voter}{c} \) that is, the group in which candidate \( c \) incurs the highest average cost.

Given {an instance $\instance$}, various cost objectives can be considered to evaluate the final winner. 
In the distributed voting, we can be even more flexible by applying different objectives at each stage. Following \citep{anshelevich2022distortion}, we consider four cost objectives:

\begin{itemize}
	\item \textbf{Average of averages} (\avgavg):  
	Average the costs within each group, then average across all groups;
	$$
	\avgavg(c \mid \instance) = \frac{1}{k} \sum_{g \in \groups} \frac{1}{\size_g} \sum_{\voter \in g} \dis{\voter}{c} = \frac{1}{k} \sum_{g \in \groups} \cost_g(c).
	$$
	
	\item \textbf{Average of maxima} (\avgmax):  
	Find the most dissatisfied voter in each group, then average their costs;
	$$
	\avgmax(c \mid \instance) = \frac{1}{k} \sum_{g \in \groups} \max_{\voter \in g} \dis{\voter}{c} = \frac{1}{k} \sum_{g \in \groups} \dis{\vstar{c}{g}}{c}.
	$$
	
	\item \textbf{Maximum of averages} (\maxavg):  
	Compute the average cost in each group; return the worst among them;
	$$
	\maxavg(c \mid \instance) = \max_{g \in \groups} ( \frac{1}{\size_g} \sum_{\voter \in g} \dis{\voter}{c}) = \cost_{\gstar(c)}(c).
	$$
	
	\item \textbf{Maximum of maxima} (\maxmax):  
	Return the cost of the most dissatisfied voter overall;
	$$
	\maxmax(c \mid \instance) = \max_{g \in \groups} \max_{\voter \in g} \dis{\voter}{c} = \dis{\vstarstar{c}}{c}.
	$$
\end{itemize}
For simplicity, when the objective is clear from context, we simply write $\cost(c \mid \instance)$. When the instance is also clear, we omit it entirely.

A voting rule $\mathsf{f}$ maps a preference profile $\profile$ to a winning candidate $\w$. This rule may be \emph{deterministic} or involve \emph{randomization}. Next, we define a distributed voting mechanism $\mech=(\fin, \fov)$, which consists of two stages:
\begin{itemize}
	\item \textbf{Stage 1:} Each group $\group$ independently selects a \emph{representative} candidate $\rep_g$ by applying the in-group rule $\fin$ to the preferences of its members. Let $R = \{\rep_g \mid \group \in \groups\}$.
	
	\item \textbf{Stage 2:} The final outcome is chosen by applying the over-group rule $\fov$ to the (centralized) instance $(R, R, \profile^R, \textsf{d})$, where $R$ acts as \emph{both the set of candidates and the set of voters}, and $\pi^R$ denotes the preferences of representatives over one another.
\end{itemize}
For each group $g$, the rule $\fin$ has access only to local information: the group size ($\size_g$) and the preference profile of its voters ($\profile^g$) over all candidates. In contrast, $\fov$ receives the preferences of the selected representatives together with the sizes of all groups.

\paragraph{Distortion.} Given a distributed mechanism \( \mech = (\fin, \fov) \), which takes an instance \( \instance = (\voters, \candidates, \groups, \profile, \textsf{d}) \) as input and outputs a winner \( \w = \mech(\instance) \), the expected cost of \( \w \) is defined as
\(\expected{\cost(\w \mid \instance)} = \sum_{g \in \groups} \Pr(\w = \rep_g) \cdot \expected{\cost(\rep_g \mid \instance)},\)
and the expected cost of \( \rep_g \) is defined as
\(\expected{\cost(\rep_g \mid \instance)} = \sum_{c \in \candidates} \Pr(\rep_g = c) \cdot \cost(c \mid \instance),\)
where \( \cost(\cdot) \) denotes one of the four objectives defined earlier. Now, we define the distortion of \( \mech \) as
\[
\distortion(\mech) = \sup_{\instance} \frac{\expected{\cost(\mech(\instance) \mid \instance)}}{\min_{c \in \candidates} \cost(c \mid \instance)}.
\]

\begin{definition}[Pareto Efficiency]\label{def:unan}
	A voting rule is \emph{pareto efficient} if, for any pair of candidates \( x \) and \( y \), if all voters prefer \( x \) to \( y \), then the rule does not select \( y \) as the winner.
\end{definition}
We now present several voting rules (for centralized settings) and distributed mechanisms (for distributed settings) frequently used in this work:

\begin{itemize}
	\item \textbf{Plurality Matching rule} (\( \fpm \)): Introduced by \citet{gkatzelis2020resolving}, this deterministic voting rule achieves a distortion of 3—the best possible among all ordinal rules in the metric setting. 
	
	\item \textbf{Plurality Matching rule with the property of Pareto Efficiency} (\( \fpmpar \)): 
	We introduce a variant of the plurality matching rule that guarantees a pareto efficient winner. Suppose there exists a candidate $c_1$ whose corresponding bipartite graph admits a perfect matching. In this bipartite graph, let one partition represent the "left" side and the other the "right" side. For each voter $\voter$ on the left, let $match_{\voter}$ denote the voter on the right matched to $\voter$. By construction, each $\voter$ prefers $c_1$ to $\topp{match_{\voter}}$. 
	
	If $c_1$ is pareto efficient, we return $c_1$ as the winner. Otherwise, there exists a candidate $c_2$ that all voters prefer to $c_1$. In this case, every voter $\voter$ prefers $c_2$ to $c_1$, and since $\voter$ also prefers $c_1$ to $\topp{match_{\voter}}$, it follows that $\voter$ prefers $c_2$ to $\topp{match_{\voter}}$. Hence, the bipartite graph corresponding to $c_2$ admits the same perfect matching as that of $c_1$.
	
	We may therefore replace $c_1$ with $c_2$. If $c_2$ is pareto efficient, it becomes the winner; otherwise, we repeat the process with $c_2$. This procedure eventually terminates with a candidate whose corresponding bipartite graph admits a perfect matching and is pareto efficient. We denote this algorithm by $\fpmpar$. The winner selected by $\fpmpar$ is both pareto efficient and associated with a perfect matching.
		
	\item \textbf{Random Dictatorship rule} (\(\frd\)): 
	A randomized voting rule in which a voter is selected uniformly at random, and the outcome is that voter's top-ranked candidate \citep{anshelevich2017randomized,kempe2020communication}. In the first stage of the \randrand\ mechanism used to establish upper bounds, we apply this rule.
	
	\item \textbf{Uniform selection rule} (\(\fur\)): 
	A randomized voting rule that selects the winner uniformly at random from the set of candidates. This rule is used in the second stage of both the \randdet\ and \randrand\ mechanisms that establish upper bounds.
	
	\item \textbf{Arbitrary Dictator mechanism} (\(\mad\)):  
	Introduced by \citet{anshelevich2022distortion}, this mechanism operates in two stages.  
	First, each group selects a representative by arbitrarily choosing a voter and taking her top-ranked candidate.  
	Second, the final winner is determined by arbitrarily selecting one of the representatives. This mechanism is employed to analyze the upper bound of the \maxmax\, objective in \detdet.
		
	\item \bm{$\alpha$}\textbf{-in-}\bm{$\beta$}\textbf{-over mechanism} ($\malphabeta$): Proposed by \citet{anshelevich2022distortion}, the $\alpha$-in-$\beta$-over mechanism operates in two deterministic stages, first applying an in-group voting rule with distortion at most $\alpha$, followed by selecting a final winner using an over-group voting rule with distortion at most $\beta$. 
\end{itemize}

\begin{definition}[Promotion]\label{promot}
	Given an order $\sigma$ over a set of candidates and a candidate $c \in \sigma$, the operation $\movetofirst{\sigma}{c}$ returns a new preference $\sigma'$ in which $c$ is moved to the top, and the relative order of all other candidates remains unchanged.
	When multiple $\movetofirst{}$ operations are applied in sequence, they are evaluated from left to right. $
	\movetofirst{\movetofirst{\sigma}{b}}{a}
	$
	first moves $b$ to the top of $\sigma$, then moves $a$ to the top of $\movetofirst{\sigma}{b}$.
\end{definition}

Using the promote operation, we define the Bias Tournament—a special complete directed graph (tournament) over the candidates—which is crucial for establishing lower bounds of the \randdet\ and \detdet\ mechanisms.

\begin{definition}[Bias Tournament]
	Let $f$ be a deterministic voting rule, $\candidates$ a set of candidates, and $\sigma$ an ordering of $\candidates$.  
	The Bias Tournament $\tour{f}{\candidates}{\sigma}$ is a complete directed graph where each vertex corresponds to a candidate in $\candidates$. For every pair of distinct candidates $u$ and $w$, there is a directed edge from $u$ to $w$ if and only if $f$ selects $u$ as the winner in a two-voter election with preferences
	$
	\profile_1 = \movetofirst{\movetofirst{\sigma}{w}}{u}  \text{ and }  \profile_2 = \movetofirst{\movetofirst{\sigma}{u}}{w}.
	$
\end{definition}

\begin{example}\label{ex:tour}
	Let \( \candidates = \{c_1, c_2, c_3\} \) and a deterministic rule \( f \) that selects the candidate with the smallest index among those ranked first by at least one voter. Suppose:  
	(i) between \( c_1 \) and \( c_2 \), the winner is \( c_1 \);  
	(ii) between \( c_1 \) and \( c_3 \), the winner is \( c_1 \);  
	(iii) between \( c_2 \) and \( c_3 \), the winner is \( c_2 \).  
	Then \( \tour{f}{\candidates}{\sigma} \) contains edges \( c_1 \to c_2 \), \( c_1 \to c_3 \), and \( c_2 \to c_3 \). Refer to \Cref{fig:tour-example} for a visual illustration.
	\begin{figure}[t]
		\centering
		\scalebox{1}{\input{figs/tour-example.tex}}
		\caption{The Bias Tournament of \Cref{ex:tour}.}
		\label{fig:tour-example}
	\end{figure}
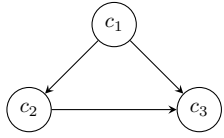
\end{example}


\subsection*{Basic Observations}
Now, we present several preliminary observations that lay the groundwork for proving our main theorems. For clarity and consistency, we fix an instance $\instance = (\voters, \candidates, \groups, \profile, \textsf{d})$ throughout this section.

\begin{observation}\label{obs:singlevoter}
	In distributed voting with single-voter groups, let $\mech = (\fin,\fov)$ be a distributed mechanism with finite distortion. Within each group, $\fin$
	must select the top-ranked candidate of each voter as the group representative.
\end{observation} 

\begin{observation}
	For any deterministic voting rule $f$ and an ordering $\sigma$ over $\candidates$, there exists a candidate with in-degree at least $\left\lceil \frac{m-1}{2} \right\rceil$ in $\tour{f}{\candidates}{\sigma}$.
	\label{pr:bias-tournamnet-indegree}
\end{observation}
\begin{proof}
	In a Bias Tournament with $m$ candidates, the sum of the in-degrees is equal to the total number of edges, which is $\binom{m}{2} = \frac{m(m-1)}{2}$. Therefore, the average in-degree is $\frac{m-1}{2}$.
	Therefore, at least one candidate must have an in-degree greater than or equal to this average. 
\end{proof}

\begin{observation} 
	Since $\opt_g$ is the optimal candidate in group $g$, we have $\cost_g(\opt_g) \le \cost_g(c)$ for any candidate $c$, including $\opt$. This holds for all objectives ($\avgmax$, $\avgavg$, $\maxmax$\, and $\maxavg$).	
	\label{prop:localcost}
\end{observation}

\begin{observation}
	Since $\cost$(.) is defined as the maximum over $\cost_g(.)$ under the $\maxavg$\, and $\maxmax$\, objectives,	 
	it follows that $\cost_g(\opt) \le \cost(\opt)$, for each group $g$.
	\label{obs:optimalcost}
\end{observation}


\begin{observation}\label{obs:expcost_randdet}
	For $\randdet$\, mechanism $\mech = (\fin, \fur)$ with output $\w $, the expected cost of the mechanism is given by \[\expected{\cost(\w)} = \frac{1}{k} \sum_{g \in \groups} \cost(\w_g).\]
\end{observation}

\begin{observation}\label{obs:expcost_randrand}
	For $\randrand$\, mechanism $\mech = (\frd, \fur)$ with output $\w $, the expected cost of the mechanism is given by \[\expected{\cost(\w)} = \frac{1}{k} \sum_{g \in \groups} \frac{1}{\size_g} \sum_{\voter \in g} \cost(\topp{\voter}).\]
\end{observation}
\begin{proof}
	By the definitions of the \emph{Random Dictatorship} rule ($\frd$) and the uniform selection rule ($\fur$), we have
	\begin{align*}
		\expected{\cost(\w)}  & = \sum_{g \in \groups} \Pr(\w = \rep_g) \cdot \expected{\cost(\rep_g)}\\
		& =\frac{1}{k} \sum_{g \in \groups} \expected{\cost(\rep_g)} \\
		& = \frac{1}{k} \sum_{g \in \groups} \sum_{\voter \in g} \Pr\left(\rep_g=\topp{\voter}\right) \cdot \cost(\topp{\voter}) \\
		& = \frac{1}{k} \sum_{g \in \groups} \frac{1}{\size_g} \sum_{\voter \in g} \cost(\topp{\voter}).
	\end{align*}
\end{proof}

\begin{observation}
	For the $\maxavg$\, objective and any group $g$, we have $\cost_{g}(\opt) \le \cost(\opt)$, as implied directly by the definition of $\maxavg$.
	\label{obs:maxavg}
\end{observation}

\begin{observation}
	Since $\topp{\voter}$ denotes the candidate closest to voter $\voter$, it follows that $\dis{\voter}{\topp{\voter}} \le \dis{\voter}{c}$ for any candidate c.
	\label{prop:topi}
\end{observation}

\begin{observation} 
	For every voter $\voter$ and every candidate $c$, we have $\dis{\voter}{c} \le \dis{\vstarstar{c}}{c}$. 
	\label{prop:maxmax}
\end{observation}

\begin{observation}
	For every group \( g \), every voter \( \voter \in g \), and every candidate \( c \), we have \( \dis{\voter}{c} \le \dis{\vstar{c}{g}}{c} \).
	\label{prop:Xmax}
\end{observation}

\begin{observation}	\label{prop:falpha}
	Consider a distributed mechanism $\mech = (\fin, \fov)$, where $\fin$ is a deterministic rule with distortion at most $\alpha$.
	By the definition of centralized distortion, we know that: 
\[
		\cost_g(\w_g) \le \alpha \cdot \cost_g(\opt_g), \quad \forall g \in \groups.
\]
\end{observation}
\begin{observation}	\label{prop:fbeta}
	Consider a $\detdet$ mechanism $\mech = (\fin, \fov)$, where $\fov$ has distortion at most $\beta$ with respect to $\avgg$ objective.
	By the definition of centralized distortion, we know that: 
	\begin{equation*}
		\frac{1}{k}\sum_{\group \in \groups}\dis{\w}{\w_\group} \le \beta \cdot \frac{1}{k}\sum_{\group \in \groups}\dis{\opt}{\w_\group}, \quad \forall{\group \in \groups}.
	\end{equation*}
\end{observation}

%% file: figs/tour-example.tex
\begin{tikzpicture}[->, >=stealth, node distance=2cm, scale=0.8, transform shape, every node/.style={font=\small}]
	
	
	\node[circle, draw, minimum size=0.1cm] (c1) {$c_1$};
	\node[circle, draw, minimum size=0.1cm, below left of=c1] (c2) {$c_2$};
	\node[circle, draw, minimum size=0.1cm, below right of=c1] (c3) {$c_3$};
	
	\draw[->] (c1) -- (c2);
	\draw[->] (c1) -- (c3);
	\draw[->] (c2) -- (c3);
	
\end{tikzpicture}

%% file: src/randdet.tex

This section examines \randdet\ mechanisms, defined as pairs $(\fin, \fov)$, where $\fin$ is a deterministic voting rule and $\fov$ is a randomized one. We establish lower and upper bounds on the distortion of these mechanisms for all cost objectives in general metric spaces.

%% file: src/randdet_upper.tex
Let $\falpha$ be a deterministic voting rule with distortion at most $\alpha$, and let $\fun$ be any deterministic voting rule that satisfies Pareto efficiency.  
We analyze the mechanisms $(\falpha, \fur)$ and $(\fun, \fur)$.  

For the $\maxavg$ and $\avgavg$ objectives, we show that the mechanism $(\falpha, \fur)$ achieves distortion at most $\alpha + 2$ and $\alpha + 2 - \nicefrac{2}{k}$, respectively.  
Since $\fpmpar$ achieves the best-known distortion of $3$ and also satisfies Pareto efficiency, we instantiate our theorems with the mechanism $(\fpmpar, \fur)$ to obtain the tightest bounds.  

Finally, for the $\avgmax$ and $\maxmax$ objectives, we prove that the mechanism $(\fun, \fur)$ achieves distortion at most $3$, a consequence of Pareto efficiency.


\begin{theorem}
	 For the $\maxavg$ objective in general metric spaces, we have $\distortion((\falpha, \fur)) \leq \alpha + 2$.
	\label{th:randdet_maxavg}
\end{theorem}
\begin{proof}
	Consider an instance $\instance=(\voters,\candidates,\groups,\profile,\textsf{d})$ and $\randdet$ mechanism $\mech = (\falpha, \fur)$. We have
	\begin{align*}
		\expected{\cost(\w)}
		&= \sum_{g \in \groups} \frac{1}{k} \cdot \cost(\rep_g) & \text{(\Cref{obs:expcost_randdet})} \\
		&= \sum_{g \in \groups} \frac{1}{k} \cdot \frac{1}{\size_{\gstar\left(\rep_g\right)}} \sum_{\voter\in \gstar\left(\rep_g\right)} \dis{\voter}{\rep_g} & \text{(Definition of $\maxavg$)} \\
		&\le \sum_{g\in \groups} \frac{1}{k} \cdot \frac{1}{\size_{\gstar\left(\rep_g\right)}} \sum_{\voter\in \gstar\left(\rep_g\right)} \bigg(\dis{\voter}{\opt} + \dis{\opt}{\rep_g}\bigg) & \text{(Triangle Inequality)}\\ 
		& = \sum_{g\in \groups} \frac{1}{k} \cdot \frac{1}{\size_{\gstar\left(\rep_g\right)}} \sum_{\voter\in \gstar\left(\rep_g\right)} \dis{\voter}{\opt} \\
		& + \sum_{g\in \groups} \frac{1}{k} \cdot \frac{1}{\size_{\gstar\left(\rep_g\right)}} \sum_{\voter\in \gstar\left(\rep_g\right)} \dis{\opt}{\rep_g} \\
		& \le \sum_{g\in \groups} \frac{1}{k} \cdot \cost(\opt) + \sum_{g\in \groups} \frac{1}{k} \cdot \frac{1}{\size_{\gstar\left(\rep_g\right)}} \sum_{\voter\in \gstar\left(\rep_g\right)} \dis{\opt}{\rep_g} & \text{(\Cref{obs:optimalcost})}
		\\ 
		& = \sum_{g\in \groups} \frac{1}{k} \cdot \cost(\opt) + \sum_{g\in \groups} \frac{1}{k} \cdot \dis{\opt}{\rep_g} 
		\\ 
		& \le \cost(\opt) + \sum_{g\in \groups} \frac{1}{k} \cdot \frac{1}{\size_{g}} \sum_{\voter\in g} \bigg(\dis{\voter}{\opt} + \dis{\voter}{\rep_g}\bigg) & \text{(Triangle Inequality)} \\
		& = \cost(\opt) + \sum_{g\in \groups} \frac{1}{k} \cdot \frac{1}{\size_{g}} \sum_{\voter\in g} \dis{\voter}{\opt} \\ 
		& + \sum_{g\in \groups} \frac{1}{k} \cdot \frac{1}{\size_{g}} \sum_{\voter\in g} \dis{\voter}{\rep_g} \\
		& \le \cost(\opt) + \sum_{g\in \groups} \frac{1}{k} \cdot \cost(\opt) +  \sum_{g\in \groups} \frac{1}{k} \cdot \cost_g(\rep_g)\\
		& = 2\cost(\opt) + \sum_{g\in \groups} \frac{1}{k} \cdot \cost_g(\rep_g) \\
		& \le 2\cost(\opt) + \sum_{g\in \groups} \frac{1}{k} \cdot \alpha \cdot \cost_g(\opt_g) &
		(\text{\Cref{prop:falpha}}) \\
		& \le 2\cost(\opt) + \alpha \sum_{g\in \groups} \frac{1}{k} \cdot \cost_g(\opt) & (\text{\Cref{prop:localcost}}) \\
		& \le 2\cost(\opt) + \alpha \sum_{g\in \groups} \frac{1}{k} \cdot \cost(\opt) & (\text{\Cref{obs:optimalcost}}) \\
		& = (\alpha + 2) \cost(\opt). 
	\end{align*}
\end{proof}

By using $\fpm$ instead of $\falpha$ as the in-group voting rule (with $\alpha = 3$) and applying \Cref{th:randdet_maxavg}, we conclude that $\mech = (\fpm, \fur)$ is a \randdet\ mechanism that satisfies the bound stated in \Cref{cor:randdet-maxavg-upper}.

\begin{corollary}[of \cref{th:randdet_maxavg}]
	For the \maxavg\, objective in general metric spaces, there exists a \randdet\, mechanism with distortion at most $5$.
	\label{cor:randdet-maxavg-upper}
\end{corollary}

\begin{theorem}
	For the \avgavg\, objective in general metric spaces, we have $\distortion((\falpha, \fur)) \leq \alpha + 2 - \tfrac{2}{k}$.
	\label{th:randdet_avgavg}
\end{theorem}
\begin{proof}
	Consider an instance $\instance=(\voters,\candidates,\groups,\profile,\textsf{d})$ and \randdet\, mechanism $\mech=(\falpha,\fur)$. We have
	\begin{align*}
		\expected{\cost(\w)}
		&= \sum_{g \in \groups} \frac{1}{k} \cdot \cost(\rep_g) & (\text{\Cref{obs:expcost_randdet}})\\
		&= \frac{1}{k} \sum_{g \in \groups} \frac{1}{k} \sum_{g' \in \groups} \cost_{g'}(\rep_g) & \text{} \\
		&= \frac{1}{k} \sum_{g \in \groups} \frac{1}{k} \sum_{g' \in \groups} \frac{1}{\size_{g'}} \sum_{\voter\in g'} \dis{\voter}{\rep_g} & (\text{Definition of }\avgavg) \\
		&\le \frac{1}{k} \sum_{g \in \groups} \frac{1}{k} \sum_{g' \in \groups, g \neq g'} \frac{1}{\size_{g'}} \sum_{\voter\in g'} \bigg(\dis{\voter}{\opt} + \dis{\opt}{\rep_g}\bigg) \\
		&+ \frac{1}{k} \sum_{g \in \groups} \frac{1}{k} \sum_{g' \in \groups, g = g'} \frac{1}{\size_{g'}} \sum_{\voter\in g'} \dis{\voter}{\rep_g} & (\text{Triangle Inequality}) \\
		&= \frac{1}{k} \sum_{g \in \groups} \frac{1}{k} \sum_{g' \in \groups, g \neq g'} \frac{1}{\size_{g'}} \sum_{\voter\in g'} \bigg(\dis{\voter}{\opt} + \dis{\opt}{\rep_g}\bigg) \\
		&+ \frac{1}{k} \sum_{g \in \groups} \frac{1}{k} \cdot \cost_g(\rep_g) \\
		&\le \frac{1}{k} \sum_{g \in \groups} \frac{1}{k} \sum_{g' \in \groups, g \neq g'} \frac{1}{\size_{g'}} \sum_{\voter\in g'} \bigg(\dis{\voter}{\opt} + \dis{\opt}{\rep_g}\bigg) \\
		&+ \frac{1}{k} \sum_{g \in \groups} \frac{1}{k} \cdot \alpha \cdot \cost_g(\opt_g) & (\text{\Cref{prop:falpha}}) \\
		&\le \frac{1}{k} \sum_{g \in \groups} \frac{1}{k} \sum_{g' \in \groups, g \neq g'} \frac{1}{\size_{g'}} \sum_{\voter\in g'} \bigg(\dis{\voter}{\opt} + \dis{\opt}{\rep_g}\bigg) \\
		&+ \frac{1}{k} \sum_{g \in \groups} \frac{1}{k} \cdot \alpha \cdot \cost_g(\opt) & (\text{\Cref{prop:localcost}}) \\
		&= \frac{1}{k} \sum_{g \in \groups} \frac{1}{k} \sum_{g' \in \groups, g \neq g'} \frac{1}{\size_{g'}} \sum_{\voter\in g'} \bigg(\dis{\voter}{\opt} + \dis{\opt}{\rep_g}\bigg) \\
		&+ \frac{\alpha}{k} \cdot \cost(\opt) & (\cost(\opt) = \sum_{g\in \groups} \frac{1}{k} \cost_g(\opt)) \\
		&= \frac{1}{k} \sum_{g \in \groups} \frac{1}{k} \sum_{g' \in \groups, g \neq g'} \frac{1}{\size_{g'}} \sum_{\voter\in g'} \dis{\voter}{\opt}  \\
		&+ \frac{1}{k} \sum_{g \in \groups} \frac{1}{k} \sum_{g' \in \groups, g \neq g'} \frac{1}{\size_{g'}} \sum_{\voter\in g'} \dis{\opt}{\rep_g} + \frac{\alpha}{k} \cdot \cost(\opt) \\
		&= \frac{1}{k} \sum_{g \in \groups} \frac{1}{k} \sum_{g' \in \groups, g \neq g'} \frac{1}{\size_{g'}} \sum_{\voter\in g'} \dis{\voter}{\opt} \\
		&+ \frac{1}{k} \sum_{g \in \groups} \frac{k-1}{k} \cdot \dis{\opt}{\rep_g} + \frac{\alpha}{k} \cdot \cost(\opt) \\
		&\le \frac{1}{k} \sum_{g \in \groups} \frac{1}{k} \sum_{g' \in \groups, g \neq g'} \frac{1}{\size_{g'}} \sum_{\voter\in g'} \dis{\voter}{\opt} \\
		&+ \frac{1}{k} \sum_{g \in \groups} \frac{k-1}{k} \cdot \frac{1}{\size_g} \sum_{\voter \in g} (\dis{\voter}{\opt}+\dis{\voter}{\rep_g}) + \frac{\alpha}{k} \cdot \cost(\opt) & (\text{Triangle Inequality}) \\
		&= \frac{1}{k} \sum_{g \in \groups} \frac{1}{k} \sum_{g' \in \groups, g \neq g'} \frac{1}{\size_{g'}} \sum_{\voter\in g'} \dis{\voter}{\opt} \\
		&+ \frac{1}{k} \sum_{g \in \groups} \frac{k-1}{k} \cdot (\cost_g(\opt) + \cost_g(\rep_g)) + \frac{\alpha}{k} \cdot \cost(\opt) & (\text{Definition of }\cost_g(.))\\
		&\le \frac{1}{k} \sum_{g \in \groups} \frac{1}{k} \sum_{g' \in \groups, g \neq g'} \frac{1}{\size_{g'}} \sum_{\voter\in g'} \dis{\voter}{\opt} \\
		&+ \frac{1}{k} \sum_{g \in \groups} \frac{k-1}{k} \cdot (\cost_g(\opt) + \alpha \cdot \cost_g(\opt_g)) + \frac{\alpha}{k} \cdot \cost(\opt) & (\text{\Cref{prop:falpha}}) \\
		&\le \frac{1}{k} \sum_{g \in \groups} \frac{1}{k} \sum_{g' \in \groups, g \neq g'} \frac{1}{\size_{g'}} \sum_{\voter\in g'} \dis{\voter}{\opt} \\
		&+ \frac{1}{k} \sum_{g \in \groups} \frac{k-1}{k} \cdot (\cost_g(\opt) + \alpha \cdot \cost_g(\opt)) + \frac{\alpha}{k} \cdot \cost(\opt) & (\text{\Cref{prop:localcost}})\\
		&= \frac{1}{k} \sum_{g \in \groups} \frac{1}{k} \sum_{g' \in \groups, g \neq g'} \frac{1}{\size_{g'}} \sum_{\voter\in g'} \dis{\voter}{\opt} \\
		&+ \frac{(k-1)(\alpha+1)}{k}  \cdot \cost(\opt) + \frac{\alpha}{k} \cdot \cost(\opt) & (\text{Definition of }\avgavg) \\
		&= \frac{1}{k} \sum_{g \in \groups} \frac{1}{k} \sum_{g' \in \groups, g \neq g'} \cost_{g'}(\opt)  + \frac{\alpha k + k -1}{k} \cdot \cost(\opt) \\
		&= \frac{1}{k} \sum_{g \in \groups} \frac{1}{k} \sum_{g' \in \groups} \cost_{g'}(\opt) - \frac{1}{k} \sum_{g \in \groups} \frac{1}{k} \sum_{g' \in \groups, g=g'} \cost_{g'}(\opt) \\ 
		&+ \frac{\alpha k + k -1}{k} \cdot \cost(\opt) \\
		&= \cost(\opt) - \frac{1}{k} \sum_{g \in \groups} \frac{1}{k} \cdot \cost_{g}(\opt) + \frac{\alpha k + k -1}{k} \cdot \cost(\opt) \\
		&= \cost(\opt) - \frac{1}{k} \cdot \cost(\opt) + \frac{\alpha k + k -1}{k} \cdot \cost(\opt) \\
		&= (\alpha + 2 - \frac{2}{k}) \cost(\opt). 
	\end{align*}
\end{proof}

Once again, by using $\fpm$ instead of $\falpha$ as the in-group voting rule (with $\alpha = 3$) and applying \Cref{th:randdet_avgavg}, we conclude that $\mech = (\fpm, \fur)$ is a \randdet\ mechanism that satisfies the bound stated in \Cref{cor:randdet-avgavg-upper}.

\begin{corollary}[of \cref{th:randdet_avgavg}]
	For the \avgavg\, objective in general metric spaces, there exists a \randdet\, mechanism with distortion at most $5- \nicefrac{2}{k}$.	
	\label{cor:randdet-avgavg-upper}
\end{corollary}

For the \avgmax\, and \maxmax\, objectives, we derive an upper bound of 5 following an argument analogous to the proof of \Cref{th:randdet_maxavg}. Now, we improve the upper bound of 5 to 3 by applying the property of pareto efficiency. 

\begin{theorem}
	\label{th:Xmaxuni}
	For the \avgmax\, and \maxmax\, objectives in general metric spaces, we have $\distortion((\fun, \fur)) \leq 3$.
\end{theorem}
\begin{proof}
	We present a proof for the \avgmax\, objective. A similar argument can be used to prove the result for the \maxmax\, objective as well. 
	
	Consider an instance $\instance = (\voters, \candidates, \groups, \profile, \textsf{d})$ and $\randdet$\ mechanism $\mech = (\fun, \fur)$. By the property of pareto efficiency, for each group $g$, there exists a voter $\voter_g \in g$ who prefers $\rep_g$ to $\opt$. Therefore, we have
	\begin{align*}
		\expected{\cost(\w)}
		& = \sum_{g \in \groups} \frac{1}{k} \cdot \cost(\rep_g) & \text{(\Cref{obs:expcost_randdet})} \\
		& = \sum_{g \in \groups} \frac{1}{k} \cdot \frac{1}{k} \sum_{g' \in \groups} \dis{\vstar{\rep_g}{g'}}{\rep_g} & \text{(Definition of $\avgmax$)} \\
		& \le \sum_{g\in \groups} \frac{1}{k} \cdot \frac{1}{k} \sum_{g' \in \groups} \bigg(\dis{\vstar{\rep_g}{g'}}{\opt} + \dis{\opt}{\rep_g}\bigg) & \text{(Triangle Inequality)} \\
		& \le \sum_{g\in \groups} \frac{1}{k} \cdot \frac{1}{k} \sum_{g' \in \groups} \bigg(\dis{\vstar{\opt}{g'}}{\opt} + \dis{\opt}{\rep_g}\bigg) & \text{(\Cref{prop:Xmax})}\\
		& = \cost(\opt) + \sum_{g \in \groups} \frac{1}{k} \cdot \dis{\opt}{\rep_g} \\
		& \le \cost(\opt) + \frac{1}{k} \sum_{g \in \groups} \bigg(\dis{\opt}{\voter_g} + \dis{\voter_g}{\rep_g}\bigg) & \text{(Triangle Inequality)} \\
		& \le \cost(\opt) + \frac{2}{k} \sum_{g \in \groups} \dis{\voter_g}{\opt} & (\dis{\voter_g}{\rep_g} \le \dis{\voter_g}{\opt}) \\
		& \le \cost(\opt) + \frac{2}{k} \sum_{g \in \groups} \dis{\vstar{\opt}{g}}{\opt} &(\text{\Cref{prop:Xmax})}\\
		& = 3\cost(\opt).
	\end{align*}
\end{proof}

%% file: src/randdet_lower.tex

Now, we establish lower bounds on the distortion of $\randdet$ mechanisms. Specifically, \Cref{th:randdet-Xmax-lower} provides lower bounds for the $\maxmax$ and $\avgmax$ objectives, \Cref{th:randdet-maxavg-lower} covers the $\maxavg$ objective, and \Cref{th:randdet-avgavg-lower} addresses the $\avgavg$ objective. It is worth noting that all the lower bounds in this section are derived from symmetric instances and apply in that setting as well. Moreover, the bounds in \Cref{th:randdet-Xmax-lower} and \Cref{th:randdet-maxavg-lower} are obtained from instances on the line metric and thus also hold in that setting.

\begin{theorem}
	For the $\avgmax$ and $\maxmax$ objectives, the distortion of any $\randdet$ mechanism is at least $3$, even when the metric space is a line.
	\label{th:randdet-Xmax-lower}
\end{theorem}
\begin{proof}
	Consider a \randdet\ mechanism $\mech=(\fin, \fov)$. We construct an instance with candidates $\candidates = \{c_1, c_2\}$ and voters $\voters = \{v_1, v_2\}$ in a single group. $c_1$ and $c_2$ are located at positions $0$ and $1$, respectively. $v_1$ and $v_2$ with preference profiles $\profile_1 = (c_1, c_2)$ and $\profile_2 = (c_2, c_1)$, are also positioned at $-0.5$ and $0.5$, respectively. Refer to \Cref{fig:randdet-Xmax} for a visual illustration. Without loss of generality, assume that $\mech$ selects $c_2$ as the representative of the group, and thus the final winner is $c_2$. Since there is only one group, the \avgmax\, and \maxmax\, objectives both simplify to \maxx.
	Thus, we have $\cost(c_1) = \frac{1}{2}$ and $\cost(c_2) = \frac{3}{2}$. Clearly, $c_1$ is the optimal candidate.
	The distortion of $\mech$ is $$\distortion(\mech) \ge \frac{\cost(c_2)}{\cost(c_1)} = 3.$$
	
	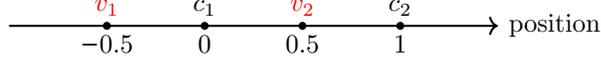
\begin{figure}[t]
		\centering
		\scalebox{1}{\input{figs/randdet-Xmax.tex}}
		\caption{An example used in the proof of \Cref{th:randdet-Xmax-lower}.}
		\label{fig:randdet-Xmax}
	\end{figure}
\end{proof}

For the \maxavg, and \avgavg, objectives, we use the Bias Tournament to establish the lower bounds stated in \Cref{th:randdet-maxavg-lower,th:randdet-avgavg-lower}.
\begin{theorem}
	For the $\maxavg$ objective, the distortion of any $\randdet$ mechanism is at least $5$, even when the metric space is a line.
	\label{th:randdet-maxavg-lower}
\end{theorem}
\begin{proof}
	Consider a $\randdet$ mechanism $\mech=(\fin, \fov)$. We construct an instance with candidates $\candidates = \{c_1, c_2, c_3, c_4\}$, and voters $\voters = \{v_1, v_2, v_3, v_4\}$, all located along a line metric. The voters are partitioned into two groups, $g_1 = \{v_1, v_2\}$ and $g_2 = \{v_3, v_4\}$.
	Let $\sigma$ be an arbitrary ordering of the candidates. Without loss of generality, assume that $c_1$ is a candidate with in-degree at least $\left\lceil \frac{m-1}{2} \right\rceil = 2$ in tournament $\tour{\fin}{\candidates}{\sigma}$, such a candidate is guaranteed to exist by
	\Cref{pr:bias-tournamnet-indegree}. Suppose $c_2$ and $c_3$ are two candidates with directed edges toward $c_1$, meaning that $c_1$ is the "losing" candidate while both $c_2$ and $c_3$ "defeat" it in the tournament.
	We may further assume that $c_2 \succ_{\sigma} c_3$. Now, consider the following construction on the line metric:
	
	\begin{itemize}
		\item Voters $v_2$ and $v_3$ are located at positions $0$, while voters $v_1$ and $v_4$ are located at $-0.5$ and $0.5$, respectively.
		
		\item Candidates $c_2$, $c_1$, and $c_3$ are located at positions $-1$, $0$, and $1$, respectively. The position of candidate $c_4$ depends on the ordering $\sigma$, ensuring the input to  $\tour{\fin}{\candidates}{\sigma}$ remains valid. We analyze three cases based on the relative ordering of $c_2, c_3, c_4$:
		\begin{itemize}
			\item \textbf{Case 1:} If $c_2 \succ_{\sigma} c_3 \succ_{\sigma} c_4$, candidate $c_4$ is located at position $10$. Refer to \Cref{fig:randdet-maxavg-lower-jkl} for an illustration.
			\begin{figure}[t]
				\centering
				\input{figs/randdet-maxavg-lower-jkl.tex}
				\caption{An example used in case 1 of \Cref{th:randdet-maxavg-lower}. Different voter groups are distinguished by distinct colors.}
				\label{fig:randdet-maxavg-lower-jkl}
			\end{figure}
			
			\item \textbf{Case 2:} If $c_2 \succ_{\sigma} c_4 \succ_{\sigma} c_3$, candidate $c_4$ is located at position $-1$. This case is illustrated in \Cref{fig:randdet-maxavg-lower-jlk}
			\begin{figure}[t]
				\centering
				\input{figs/randdet-maxavg-lower-jlk.tex}
				\caption{An example used in case 2 of \Cref{th:randdet-maxavg-lower}. Different voter groups are distinguished by distinct colors.}
				\label{fig:randdet-maxavg-lower-jlk}
			\end{figure}
			
			\item \textbf{Case 3:} If $c_4 \succ_{\sigma} c_2 \succ_{\sigma} c_3$, candidate $c_4$ is located at $1$. A visual representation of this case can be found in \Cref{fig:randdet-maxavg-lower-ljk}
			\begin{figure}[t]
				\centering
				\input{figs/randdet-maxavg-lower-ljk.tex}
				\caption{An example used in case 3 of \Cref{th:randdet-maxavg-lower}. 
				Different voter groups are distinguished by distinct colors.}
				\label{fig:randdet-maxavg-lower-ljk}
			\end{figure}
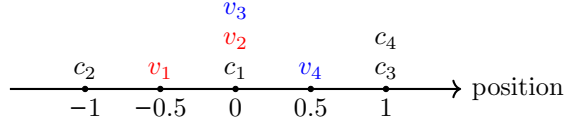
		\end{itemize}
	\end{itemize}
	Note that when a voter is equidistant from two candidates, multiple preference profiles may be consistent with the underlying metric space. According to $\tour{\fin}{\candidates}{\sigma}$, we can determine the group representatives:
	\begin{itemize}
		\item A directed edge from $c_2$ to $c_1$, implies that $c_2$ is selected as the representative for group $g_1$.
		
		\item Similarly, a directed edge from $c_3$ to $c_1$, means $c_3$ is the representative for group $g_2$.
	\end{itemize}
	By the definition of the $\maxavg$ objective, we have $\cost(c_2) = \cost(c_3) = \frac{5}{4}$ and $\cost(c_1) = \frac{1}{4}$. Thus, $c_1$ is the optimal candidate in all cases.
 	The mechanism must select the final winner from the group representatives, $c_2$ or $c_3$. Finally, we derive the distortion of mechanism $\mech$ as:
	\begin{align*}
		\distortion(\mech)
		& \ge \min \left({\frac{\cost(c_2)}{\cost(\opt)},\frac{\cost(c_3)}{\cost(\opt)}} \right) \\
		& = \frac{\frac{5}{4}}{\cost(c_1)} \\
		& = 5. 
	\end{align*}
\end{proof}

We now establish the lower bound for the \avgavg\ objective in \Cref{th:randdet-avgavg-lower}.

\begin{theorem}
	For general metric spaces and the \avgavg\, objective, the distortion of any $\randdet$ mechanism is at least $5-\frac{2}{k}$.
	\label{th:randdet-avgavg-lower}
\end{theorem}


\begin{proof}
	Consider a \randdet\ mechanism $\mech=(\fin, \fov)$. We construct an instance with a set of $m=2k$ candidates, $\candidates = \{c_1, c_2, \ldots, c_{m = 2k}\}$, a set of $n=2k$ voters, $\voters = \{v_1, v_2, \ldots, v_{n = 2k}\}$, and $k$ groups $g_i = \{v_{2i-1}, v_{2i}\}$ for $1 \le i \le k$.
	Let $\sigma$ be an arbitrary ordering of the candidates. Without loss of generality, assume that $c_{2k}$ is a candidate with in-degree at least $\left\lceil \frac{m - 1}{2} \right\rceil = k$ in the tournament $\tour{\fin}{\candidates}{\sigma}$, such a candidate is guaranteed to exist by \Cref{pr:bias-tournamnet-indegree}. Further, suppose $c_1, c_2, \ldots, c_k$ are $k$ candidates that have directed edges toward $c_{2k}$ in the tournament, meaning that $c_{2k}$ is the "losing" candidate.
		
	We construct a connected graph $G$ with $2k+3$ vertices, denoted $u_1, u_2, \ldots, u_{2k+3}$, where the shortest-path distances in $G$ define the underlying metric space $\textsf{d}$. Each voter and candidate is placed on one of the vertices (a single vertex may host multiple entities). The graph $G$ is constructed as follows (see \Cref{fig:randdet-avgavg-graph} for an illustration):
	\begin{itemize}
		\item Place candidate \( c_{2k} \) at vertex \( u_1 \).
		
		\item For each \( 1 \le i \le k+1 \), add an edge between \( u_1 \) and \( u_{2i} \), and another edge between \( u_{2i} \) and \( u_{2i+1} \). This forms \( k+1 \) branches extending from the central vertex \( u_1 \).
		
		\item For each \( 1 \le i \le k \), place voter \( v_{2i-1} \) at vertex \( u_1 \), voter \( v_{2i} \) at vertex \( u_{2i} \), and candidate \( c_i \) at vertex \( u_{2i+1} \).
		
		\item For each \( k+1 \le i \le 2k-1 \), place candidate \( c_i \) at vertex \( u_{2k+3} \).
	\end{itemize}	

\begin{figure}[t]
	\centering
	{\input{figs/randdet-avgavg-graph.tex}}
	\caption{Tree graph used in the proof of \Cref{th:randdet-avgavg-lower}. Different voter groups are distinguished by distinct colors.}
	\label{fig:randdet-avgavg-graph}
\end{figure}
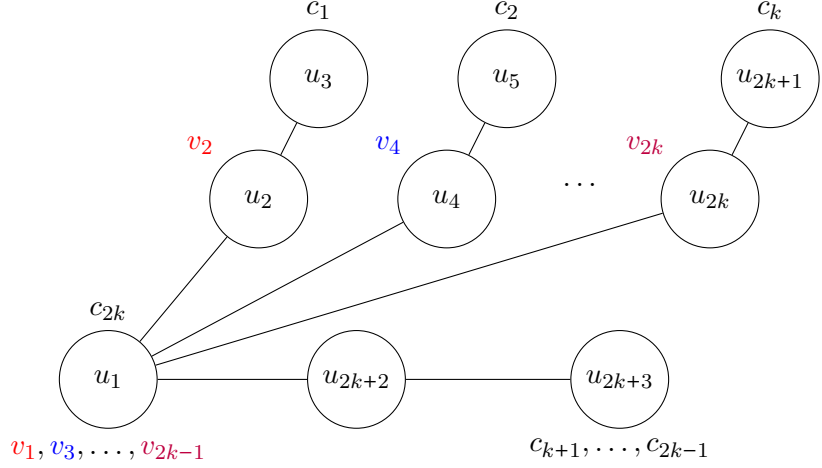

Pairwise distances between the candidates and voters are presented in \Cref{tab:randdet-avgavg-table-x,tab:randdet-avgavg-table-y}. Moreover, the preference profiles in \Cref{tab:randdet-avgavg-profiles} induced by the shortest-path distances in graph $G$, are consistent with the metric space $\textsf{d}$. Note that multiple preference profiles may be consistent with $\textsf{d}$.

According to $\tour{\fin}{\candidates}{\sigma}$, the representative of group $g_i$ is candidate $c_i$ for any $1 \leq i \leq k$. Thus, the mechanism must select one of these representatives as the final winner. By the definition of the $\avgavg$ objective, we have $\cost(c_{2k}) = \frac{1}{2}$, and $\cost(c_i) = \frac{5 - \frac{2}{k}}{2}$, for all $1 \leq i \leq k$. Thus, $c_{2k}$ is the optimal candidate. It follows that the distortion of mechanism $\mech$:

\begin{align*}
	\distortion(\mech) & \ge \frac{\min_{1 \le i \le k} \bigg(\cost(c_{i}) \bigg)}{\cost(\opt)} \\
	& = \frac{\frac{5 - \frac{2}{k}}{2}}{\cost(c_{2k})} \\
	& = 5 - \frac{2}{k}.
\end{align*}
%

\begin{table}[t]
	\centering
	\begin{tabular}{|>{\columncolor{gray!20}}c|*{1}{>{\centering\arraybackslash}p{1cm}}|}
		\hline
		\rowcolor{gray!20}
		$\dis{\cdot}{\cdot}$ & $\bm{c_i}$ \\
		\hline
		$\bm{v_{2i-1}}$ & 2 \\
		$\bm{v_{2i}}$ & 1 \\
		$\bm{v_{2j-1}}$ & 2 \\
		$\bm{v_{2j}}$ & 3 \\
		\hline
	\end{tabular}
	\vspace{0.3cm}
	\caption{For any \(1 \le i, j \le k\) with \(i \ne j\), the shortest-path distances in graph $G$ between candidates \(c_1, c_2, \ldots, c_k\) and the voters used in the proof of \Cref{th:randdet-avgavg-lower}.}
	\label{tab:randdet-avgavg-table-x}
\end{table}
\begin{table}[t]
	\centering
	\begin{tabular}{|>{\columncolor{gray!20}}c|*{2}{>{\centering\arraybackslash}p{1cm}}|}
		\hline
		\rowcolor{gray!20}
		$\dis{\cdot}{\cdot}$ & $\bm{c_i}$ & $\bm{c_{2k}}$\\
		\hline
		$\bm{v_{2j-1}}$ & 2 & 0\\
		$\bm{v_{2j}}$ & 3 & 1 \\
		\hline
	\end{tabular}
	\vspace{0.3cm}
	\caption{For any \(k+1 \le i \le 2k-1\) and \(1 \le j \le k\), the shortest-path distances in graph $G$ between candidates \(c_{k+1}, c_{k+2}, \ldots, c_{2k-1}, c_{2k}\) and the voters used in the proof of  \cref{th:randdet-avgavg-lower}.}
	\label{tab:randdet-avgavg-table-y}
\end{table}
\begin{table}[t]
	\centering
	\begin{tabular}{|>{\columncolor{gray!20}}c|*{1}{>{\centering\arraybackslash}p{3.5cm}}|}
		\hline
		\rowcolor{gray!20}
		Voter & Preference Profile \\
		\hline
		$\bm{v_{2i-1}}$ & $\movetofirst{\movetofirst{\sigma}{c_i}}{c_{2k}}$ \\
		$\bm{v_{2i}}$ & $\movetofirst{\movetofirst{\sigma}{c_{2k}}}{c_i}$ \\
		\hline
	\end{tabular}
	\vspace{0.3cm}
	\caption{The preference profiles of the voters within each group $g_i$ 
	\( (1 \le i \le k) \), used in the proof of \cref{th:randdet-avgavg-lower}.}
	\label{tab:randdet-avgavg-profiles}
\end{table}

\end{proof}

%% file: figs/randdet-Xmax.tex
\begin{tikzpicture}[x=2cm, y=1cm, font=\small, scale=1.3]
	
	\draw[->, thick] (-1, 0) -- (1.5, 0) node[right] {position};

	\coordinate (V1) at (-0.5,0);
	\coordinate (C1) at (0,0);
	\coordinate (OTV) at (0.5,0);
	\coordinate (OTC) at (1,0);
	
	\filldraw[black] (V1) circle (1pt) node[anchor=south] {\textcolor{red}{$v_1$}};
	\filldraw[black] (C1) circle (1pt) node[anchor=south] {\textcolor{black}{$c_1$}};
	\filldraw[black] (OTV) circle (1pt) node[anchor=south] {\textcolor{red}{$v_2$}};
	\filldraw[black] (OTC) circle (1pt) node[anchor=south] {\textcolor{black}{$c_2$}};
	
	\foreach \x/\lab in {-0.5/{$-0.5$}, 0/0, 0.5/{$0.5$}, 1/{$1$}} {
		\node[below] at (\x, 0) {\lab};
	}
	
\end{tikzpicture}

%% file: figs/randdet-maxavg-lower-jkl.tex
\begin{tikzpicture}[x=2cm, y=1cm, font=\small]
	
	\draw[->, thick] (-1.5, 0) -- (5, 0) node[right] {position};
	
	\coordinate (V1) at (-0.5,0);
	\coordinate (V2) at (0,0.4);
	\coordinate (V3) at (0,0.8); 
	\coordinate (V4) at (0.5,0);
	\coordinate (C1) at (0,0);
	\coordinate (C2) at (-1,0);
	\coordinate (C3) at (1,0);
	\coordinate (C4) at (4,0);
	
	\filldraw[black] (V1) circle (1pt) node[anchor=south] {\textcolor{red}{$v_1$}};
	\filldraw[black] (V2) node[anchor=south] {\textcolor{red}{$v_2$}};
	\filldraw[black] (V3) node[anchor=south] {\textcolor{blue}{$v_3$}};
	\filldraw[black] (V4) circle (1pt) node[anchor=south] {\textcolor{blue}{$v_4$}};
	
	\filldraw[black] (C1) circle (1pt) node[anchor=south] {\textcolor{black}{$c_1$}};
	\filldraw[black] (C2) circle (1pt) node[anchor=south] {\textcolor{black}{$c_2$}};
	\filldraw[black] (C3) circle (1pt) node[anchor=south] {\textcolor{black}{$c_3$}};
	\filldraw[black] (C4) circle (1pt) node[anchor=south] {\textcolor{black}{$c_4$}};
	
	\foreach \x/\lab in {-1/$-1$, -0.5/{$-0.5$}, 0/0, 0.5/{$0.5$}, 1/{$1$}} {
		\node[below] at (\x, 0) {\lab};
	}
	\foreach \x/\lab in {10/{$10$}} {
		\node[below] at (4, 0) {\lab};
	}
\end{tikzpicture}

%% file: figs/randdet-maxavg-lower-jlk.tex
\begin{tikzpicture}[x=2cm, y=1cm, font=\small]
	
	\draw[->, thick] (-1.5, 0) -- (1.5, 0) node[right] {position};
	
	\coordinate (V1) at (-0.5,0);
	\coordinate (V2) at (0,0.4);
	\coordinate (V3) at (0,0.8); 
	\coordinate (V4) at (0.5,0);
	\coordinate (C1) at (0,0);
	\coordinate (C2) at (-1,0);
	\coordinate (C3) at (1,0);
	\coordinate (C4) at (-1,0.4);
	
	\filldraw[black] (V1) circle (1pt) node[anchor=south] {\textcolor{red}{$v_1$}};
	\filldraw[black] (V2) node[anchor=south] {\textcolor{red}{$v_2$}};
	\filldraw[black] (V3) node[anchor=south] {\textcolor{blue}{$v_3$}};
	\filldraw[black] (V4) circle (1pt) node[anchor=south] {\textcolor{blue}{$v_4$}};
	
	\filldraw[black] (C1) circle (1pt) node[anchor=south] {\textcolor{black}{$c_1$}};
	\filldraw[black] (C2) circle (1pt) node[anchor=south] {\textcolor{black}{$c_2$}};
	\filldraw[black] (C3) circle (1pt) node[anchor=south] {\textcolor{black}{$c_3$}};
	\filldraw[black] (C4) node[anchor=south] {\textcolor{black}{$c_4$}};
	\foreach \x/\lab in {-1/$-1$, -0.5/{$-0.5$}, 0/0, 0.5/{$0.5$}, 1/{$1$}} {
		\node[below] at (\x, 0) {\lab};
	}
\end{tikzpicture}

%% file: figs/randdet-maxavg-lower-ljk.tex
\begin{tikzpicture}[x=2cm, y=1cm, font=\small]
	
	\draw[->, thick] (-1.5, 0) -- (1.5, 0) node[right] {position};
	
	\coordinate (V1) at (-0.5,0);
	\coordinate (V2) at (0,0.4);
	\coordinate (V3) at (0,0.8); 
	\coordinate (V4) at (0.5,0);
	\coordinate (C1) at (0,0);
	\coordinate (C2) at (-1,0);
	\coordinate (C3) at (1,0);
	\coordinate (C4) at (1,0.4);
	
	\filldraw[black] (V1) circle (1pt) node[anchor=south] {\textcolor{red}{$v_1$}};
	\filldraw[black] (V2) node[anchor=south] {\textcolor{red}{$v_2$}};
	\filldraw[black] (V3) node[anchor=south] {\textcolor{blue}{$v_3$}};
	\filldraw[black] (V4) circle (1pt) node[anchor=south] {\textcolor{blue}{$v_4$}};
	
	\filldraw[black] (C1) circle (1pt) node[anchor=south] {\textcolor{black}{$c_1$}};
	\filldraw[black] (C2) circle (1pt) node[anchor=south] {\textcolor{black}{$c_2$}};
	\filldraw[black] (C3) circle (1pt) node[anchor=south] {\textcolor{black}{$c_3$}};
	\filldraw[black] (C4) node[anchor=south] {\textcolor{black}{$c_4$}};
	\foreach \x/\lab in {-1/$-1$, -0.5/{$-0.5$}, 0/0, 0.5/{$0.5$}, 1/{$1$}} {
		\node[below] at (\x, 0) {\lab};
	}
\end{tikzpicture}

%% file: figs/randdet-avgavg-graph.tex
\begin{tikzpicture}[every node/.style={}, 
	label distance=0pt
	]
	
	\node[minimum size=1.3cm, circle, draw] (u1) at (0,0) {$u_1$};
	\node[above=0pt of u1, draw=none] {$c_{2k}$};
	\node[below=1pt of u1, draw=none] {$\textcolor{red}{v_1}, \textcolor{blue}{v_3}, \dots, \textcolor{purple}{v_{2k-1}}$};
	
	\node[minimum size=1.3cm, circle, draw] (u2) at (2,2.4) {$u_2$};
	\node[above left=0pt of u2, draw=none] {\textcolor{red}{$v_2$}};
	
	\node[minimum size=1.3cm, circle, draw] (u3) at (2.8,4.0) {$u_3$};
	\node[above=0pt of u3, draw=none] {$c_1$};
	\draw (u1) -- (u2) -- (u3);
	
	\node[minimum size=1.3cm, circle, draw] (u4) at (4.5,2.4) {$u_4$};
	\node[above left=0pt of u4, draw=none] {\textcolor{blue}{$v_4$}};
	
	\node[minimum size=1.3cm, circle, draw] (u5) at (5.3,4.0) {$u_5$};
	\node[above=0pt of u5, draw=none] {$c_2$};
	\draw (u1) -- (u4) -- (u5);
	
	\node[draw=none] at (6.3, 2.6) {$\dots$};
	
	\node[minimum size=1.3cm, circle, draw] (uN1) at (8,2.4) {$u_{2k}$};
	\node[above left=0pt of uN1, draw=none] {\textcolor{purple}{$v_{2k}$}};
	
	\node[minimum size=1.3cm, circle, draw] (uN2) at (8.8,4.0) {$u_{2k+1}$};
	\node[above=0pt of uN2, draw=none] {$c_k$};
	\draw (u1) -- (uN1) -- (uN2);
	
	\node[minimum size=1.3cm, circle, draw] (uy) at (6.8,0) {$u_{2k+3}$};
	\node[below=0pt of uy, draw=none] {$c_{k+1},\dots,c_{2k-1}$};
	\node[minimum size=1.3cm, circle, draw] (uM2) at (3.3,0) {$u_{2k+2}$};
	\draw (u1) -- (uM2) -- (uy);
	
\end{tikzpicture}

%% file: src/randrand.tex
This section examines \randrand\ mechanisms, which are pairs $(\fin, \fov)$ composed of two independently randomized voting rules $\fin$ and $\fov$. We establish lower and upper bounds on the distortion of these mechanisms for all cost objectives in general metric spaces.

%% file: src/randrand_upper.tex

Throughout this section, we analyze the mechanism $(\frd,\fur)$ and show that despite its simplicity, it achieves tight (or nearly tight) distortion bounds for various cost objectives.
In particular, for the \maxmax\, objective (indeed the $\maxx$ objective), we establish that choosing the top candidate of any voter yields a distortion of at most 3. 
We begin with the simplest case, \maxmax, and move towards the most intricate \avgavg.
\begin{theorem}
	For the \maxmax\, objective in general metric spaces, we have $\distortion((\frd,\fur)) \leq 3$.
	\label{th:randrand_maxmax}
\end{theorem}
\begin{proof}
	Consider an instance $\instance=(\voters,\candidates,\groups,\profile,\textsf{d})$ and \randrand\, mechanism $\mech=(\frd,\fur)$. For any voter $\voter \in \voters$, we have
	\begin{align*}
		\cost(\topp{\voter}) & = \dis{\vstarstar{\topp{\voter}}}{\topp{\voter}} & (\text{Definition of \maxmax}) \\
		& \le \dis{\vstarstar{\topp{\voter}}}{\opt} + \dis{\opt}{\topp{\voter}} & (\text{Triangle Inequality}) \\
		& \le \dis{\vstarstar{\topp{\voter}}}{\opt} + \dis{\voter}{\opt} + \dis{\voter}{\topp{\voter}} & (\text{Triangle Inequality}) \\
		& \le \dis{\vstarstar{\opt}}{\opt} + \dis{\voter}{\opt} + \dis{\voter}{\topp{\voter}} & (\text{\Cref{prop:maxmax}}) \\
		& \le \dis{\vstarstar{\opt}}{\opt} + \dis{\voter}{\opt} + \dis{\voter}{\opt} & (\text{\Cref{prop:topi}}) \\
		& = \dis{\vstarstar{\opt}}{\opt} + 2\dis{\voter}{\opt} \\
		& \le 3\dis{\vstarstar{\opt}}{\opt} & (\text{\Cref{prop:maxmax}}) \\
		& =3\cost(\opt) & (\cost(\opt) = \dis{\vstarstar{\opt}}{\opt}).
	\end{align*}
	Combining this with \Cref{obs:expcost_randrand}, we have
	\begin{align*}
		\expected{\cost(\w)} & = \frac{1}{k} \sum_{g \in \groups} \frac{1}{\size_g} \sum_{\voter \in g} \cost(\topp{\voter}) \\
		& \le 3\cost(\opt).
	\end{align*}	
\end{proof}

\begin{theorem}
	For the \avgmax\, objective in general metric spaces, we have $\distortion((\frd,\fur)) \leq 3$.
	\label{th:randrand_avgmax}
\end{theorem}
\begin{proof}
	Consider an instance $\instance=(\voters,\candidates,\groups,\profile,\textsf{d})$ and \randrand\, mechanism $\mech = (\frd,\fur)$. By definition of the \avgmax\, objective for any voter $\voter \in \voters$, we have 
	$$\cost(\topp{\voter}) = \frac{1}{k} \sum_{g \in \groups} \cost_{g}(\topp{\voter}).$$
	Now, for any groups $g, g' \in \groups$ and any voter $\voter \in g'$, we have
	\begin{align*}
		\cost_{g}(\topp{\voter}) & = \dis{\vstar{\topp{\voter}}{g}}{\topp{\voter}} & (\text{Definition of $\cost_{g}(.)$}) \\
		& \le \dis{\vstar{\topp{\voter}}{g}}{\voter} + \dis{\voter}{\topp{\voter}} & (\text{Triangle Inequality}) \\
		& \le \dis{\vstar{\topp{\voter}}{g}}{\voter} + \dis{\voter}{\opt} & (\text{\Cref{prop:topi}}) \\
		& \le \dis{\vstar{\topp{\voter}}{g}}{\opt} + \dis{\voter}{\opt} + \dis{\voter}{\opt} & (\text{Triangle Inequality}) \\
		& \le \dis{\vstar{\opt}{g}}{\opt} + 2\dis{\voter}{\opt} & (\text{\Cref{prop:Xmax}}) \\
		& = \cost_{g}(\opt) + 2\dis{\voter}{\opt} & (\text{Definition of $\cost_{g}(.)$}) \\
		& \le \cost_{g}(\opt) + 2\dis{\vstar{\opt}{g'}}{\opt} & (\text{\Cref{prop:Xmax}}) \\
		& =\cost_{g}(\opt) + 2\cost_{g'}(\opt) & (\text{Definition of $\cost_{g}(.)$}).
	\end{align*}
	Combining this with \Cref{obs:expcost_randrand}, we obtain
	\begin{align*}
		\expected{\cost(\w)} &= \frac{1}{k} \sum_{g \in \groups} \frac{1}{\size_g} \sum_{\voter \in g} \frac{1}{k} \sum_{g' \in \groups} \cost_{g'}(\topp{\voter}) \\
		&\le \frac{1}{k} \sum_{g \in \groups} \frac{1}{\size_g} \sum_{\voter \in g} \frac{1}{k} \sum_{g' \in \groups} \left( 2 \cost_{g}(\opt) + \cost_{g'}(\opt) \right) \\
		&= \frac{1}{k} \sum_{g \in \groups} \left( 2 \cost_{g}(\opt) + \frac{1}{k} \sum_{g' \in \groups} \cost_{g'}(\opt) \right) \\
		&= \frac{1}{k} \sum_{g \in \groups} 2 \cost_{g}(\opt) + \cost(\opt) & (\text{Definition of $\cost(.)$}) \\
		&= 3 \cost(\opt) & (\text{Definition of $\cost(.)$}).
	\end{align*}
\end{proof}

For the \maxavg\, objective, the key insight is to show that for any voter $\voter \in \voters$,  $\cost_{\gstar\left(\topp{\voter}\right)}(\topp{\voter}) \leq  2\dis{\voter}{\opt} + \cost(\opt)$. This crucial inequality is the foundation for proving the desired upper bound.

\begin{theorem}
	For the \maxavg\, objective in general metric spaces, we have $\distortion((\frd,\fur)) \leq 3$.
	\label{th:randrand_maxavg}
\end{theorem}
\begin{proof}
	Consider an instance $\instance=(\voters,\candidates,\groups,\profile,\textsf{d})$ and \randrand\, mechanism $\mech = (\frd,\fur)$. By the  definition of the \maxavg\, objective for any voter $\voter \in \voters$, we have 
	\begin{align*}
		\cost(\topp{\voter}) &= \max_{g \in \groups} \cost_{g}(\topp{\voter}) \\
		&= \cost_{\gstar\left(\topp{\voter}\right)}(\topp{\voter}) \\
		&= \frac{1}{\size_{\gstar\left(\topp{\voter}\right)}} \sum_{\voter' \in \gstar\left(\topp{v}\right)} \dis{\voter'}{\topp{\voter}} & (\text{Definition of }\cost_g(.)).
	\end{align*}
	Thus, we have 
	\begin{align*}
		\cost_{\gstar\left(\topp{\voter}\right)}(\topp{\voter}) & = \frac{1}{\size_{\gstar\left(\topp{\voter}\right)}} \sum_{\voter' \in \gstar\left(\topp{\voter}\right)} \dis{\voter'}{\topp{\voter}} \\
		& \le \frac{1}{\size_{\gstar\left(\topp{\voter}\right)}} \sum_{\voter' \in \gstar\left(\topp{\voter}\right)} \bigg(\dis{\voter}{\topp{\voter}} + \dis{\voter}{\voter'}\bigg) & (\text{Triangle Inequality}) \\
		& \le \frac{1}{\size_{\gstar\left(\topp{\voter}\right)}} \sum_{\voter' \in \gstar\left(\topp{\voter}\right)} \bigg(\dis{\voter}{\opt} + \dis{\voter}{\voter'}\bigg) & (\text{\Cref{prop:topi}}) \\
		& \le \frac{1}{\size_{\gstar\left(\topp{\voter}\right)}} \sum_{\voter' \in \gstar\left(\topp{\voter}\right)} \bigg(\dis{\voter}{\opt} + \dis{\voter}{\opt} + \dis{\opt}{\voter'} \bigg) & (\text{Triangle Inequality}) \\
		& = 2\dis{\voter}{\opt} + \frac{1}{\size_{\gstar\left(\topp{\voter}\right)}} \sum_{\voter' \in \gstar\left(\topp{\voter}\right)} \dis{\opt}{\voter'} \\
		& = 2\dis{\voter}{\opt} + \cost_{\gstar\left(\topp{\voter}\right)}(\opt)\\
		& \le 2\dis{\voter}{\opt} + \cost(\opt) & (\text{\Cref{obs:maxavg}}).
	\end{align*}
	Combining this with \Cref{obs:expcost_randrand}, we obtain
	\begin{align*}
		\expected{\cost(\w)} &= \frac{1}{k} \sum_{g \in \groups} \frac{1}{\size_g} \sum_{\voter \in g} \cost(\topp{\voter}) \\
		& = \frac{1}{k} \sum_{g \in \groups} \frac{1}{\size_g} \sum_{\voter \in g} \cost_{\gstar\left(\topp{\voter}\right)}(\topp{\voter}) & (\text{Definition of }\cost(.)) \\
		& \le \frac{1}{k} \sum_{g \in \groups} \frac{1}{\size_g} \sum_{\voter \in g} \left( 2\dis{\voter}{\opt} + \cost(\opt) \right)\\
		& = \cost(\opt) + \frac{1}{k} \sum_{g \in \groups} \frac{1}{\size_g} \sum_{\voter \in g} 2\dis{\voter}{\opt}\\
		& = \cost(\opt) + \frac{2}{k} \sum_{g \in \groups} \cost_g(\opt) & (\text{Definition of }\cost_g(.)) \\
		& \le \cost(\opt) + \frac{2}{k} \sum_{g \in \groups} \cost(\opt) & (\text{\Cref{obs:maxavg}})\\
		& = 3\cost(\opt).
	\end{align*}
\end{proof}

\begin{theorem}
	For the \avgavg\, objective in general metric spaces, we have $\distortion((\frd,\fur)) \leq 3 - \nicefrac{2}{k\size^*}$ where $\size^*$ represents the maximum value of $\size_g$ over all groups.
	\label{th:randrand_avgavg}
\end{theorem}
\begin{proof}
	Consider an instance $\instance=(\voters,\candidates,\groups,\profile,\textsf{d})$ and $\randrand$ mechanism $\mech=(\frd,\fur)$. By the  definition of the \avgavg\, objective for any voter $\voter \in \voters$, we have 
	$$\cost(\topp{\voter}) = \frac{1}{k} \sum_{g' \in \groups} \cost_{g'}(\topp{\voter}).$$
	For any voter $\voter \in \voters$ and group $g'$, we have 
	\begin{align*}
		\cost_{g'}(\topp{\voter}) & = \frac{1}{\size_{g'}} \sum_{\voter' \in g'} \dis{\voter'}{\topp{\voter}} & (\text{Definition of }\cost_g(.)) \\
		&\le \frac{1}{\size_{g'}} \sum_{\voter' \in g'} \bigg(\dis{\voter}{\topp{\voter}} + \dis{\voter'}{\voter} \bigg) & (\text{Triangle Inequality}) \\
		&\le \frac{1}{\size_{g'}} \sum_{\voter' \in g'} \bigg(\dis{\voter}{\opt} + \dis{\voter'}{\voter} \bigg) & (\text{\Cref{prop:topi}}) \\
		&= \dis{\voter}{\opt} + \frac{1}{\size_{g'}} \sum_{\voter' \in g'} \dis{\voter'}{\voter}. 
	\end{align*}
	Combining this with \Cref{obs:expcost_randrand}, we obtain
	\begin{align*}
		\expected{\cost(\w)} &= \frac{1}{k} \sum_{g \in \groups} \frac{1}{\size_g} \sum_{\voter \in g} \frac{1}{k} \sum_{g' \in \groups} \cost_{g'}(\topp{\voter}) & (\text{Definition of }\cost(.))\\
		& \le \frac{1}{k} \sum_{g \in \groups} \frac{1}{\size_g} \sum_{\voter \in g} \frac{1}{k} \sum_{g' \in \groups} \bigg(\dis{\voter}{\opt} + \frac{1}{\size_{g'}} \sum_{\voter' \in g'} \dis{\voter'}{\voter}\bigg) \\
		& = \frac{1}{k} \sum_{g \in \groups} \frac{1}{\size_g} \sum_{\voter \in g} \dis{\voter}{\opt} + \frac{1}{k} \sum_{g \in \groups} \frac{1}{\size_g} \sum_{\voter \in g} \frac{1}{k} \sum_{g' \in \groups} \frac{1}{\size_{g'}} \sum_{\voter' \in g'} \dis{\voter'}{\voter} \\
		& = \cost(\opt) + \frac{1}{k} \sum_{g \in \groups} \frac{1}{\size_g} \sum_{\voter \in g} \frac{1}{k} \sum_{g' \in \groups} \frac{1}{\size_{g'}} \sum_{\voter' \in g'} \dis{\voter'}{\voter} &(\text{Definition of \avgavg})\\
		& \le \cost(\opt) + \frac{1}{k} \sum_{g \in \groups} \frac{1}{\size_g} \sum_{\voter \in g} \frac{1}{k} \sum_{g' \in \groups} \frac{1}{\size_{g'}} \sum_{\voter' \in g',\,\voter' \ne \voter} \bigg( \dis{\voter}{\opt} + \dis{\opt}{\voter'} \bigg) &(\text{Triangle Inequality})\\
		& = \cost(\opt) + \frac{1}{k} \sum_{g \in \groups} \frac{1}{\size_g} \sum_{\voter \in g} \frac{1}{k} \sum_{g' \in \groups} \frac{1}{\size_{g'}} \sum_{\voter' \in g'} \bigg( \dis{\voter}{\opt} + \dis{\opt}{\voter'} \bigg) \\
		& - \frac{1}{k} \sum_{g \in \groups} \frac{1}{\size_g} \sum_{\voter \in g} \frac{1}{k} \sum_{g' \in \groups, g'=g} \frac{1}{\size_{g'}} \sum_{\voter' \in g', \voter'=\voter} \bigg( \dis{\voter}{\opt} + \dis{\opt}{\voter'} \bigg) \\
		& = \cost(\opt) + \frac{1}{k} \sum_{g \in \groups} \frac{1}{\size_g} \sum_{\voter \in g} \frac{1}{k} \sum_{g' \in \groups} \frac{1}{\size_{g'}} \sum_{\voter' \in g'} \bigg( \dis{\voter}{\opt} + \dis{\opt}{\voter'} \bigg) \\
		& - \frac{1}{k} \sum_{g \in \groups} \frac{1}{\size_g} \sum_{\voter \in g} \frac{1}{k} \cdot \frac{1}{\size_{g}} \cdot 2\dis{\voter}{\opt} \\
		& = \cost(\opt) + 2\cost(\opt) - \frac{1}{k} \sum_{g \in \groups} \frac{1}{\size_g} \sum_{\voter \in g} \frac{1}{k} \cdot \frac{1}{\size_{g}} \cdot 2\dis{\voter}{\opt} &(\text{Definition of \avgavg})\\
		& = 3\cost(\opt) - \frac{1}{k} \sum_{g \in \groups} \frac{1}{\size_g}
		\sum_{\voter \in g} \frac{1}{k} \cdot \frac{1}{\size_{g}} \cdot 2\dis{\voter}{\opt}
		\\
		& = 3\cost(\opt) - \frac{2}{k^2} \sum_{g \in \groups} \frac{1}{\size_g} \cdot \cost_{g}(\opt) & (\text{Definition of $\cost_{g}(.)$})\\
		& \le 3\cost(\opt) - \frac{2}{k^2} \sum_{g \in \groups} \frac{1}{\size^*} \cdot \cost_{g}(\opt) \\
		& = (3 - \frac{2}{k\size^*})\cost(\opt) &(\text{Definition of \avgavg}).
	\end{align*}
\end{proof}

As a corollary of \Cref{th:randrand_avgavg}, 
we conclude that $\mech=(\frd,\fur)$ is a \randrand\, mechanism that satisfies the bound stated in \Cref{cor:randrand-avgavg-upper}, particularly in the symmetric case $(kn^* = n)$.
\begin{corollary}[of \cref{th:randrand_avgavg}]
	For the \avgavg\, objective in general metric spaces, there exists a \randrand\, mechanism with distortion
	at most \( 3 - \tfrac{2}{n} \), provided that the groups are symmetric.
	\label{cor:randrand-avgavg-upper}
\end{corollary}

%% file: src/randrand_lower.tex


All bounds in this section are derived from the symmetric instances and thus apply to that setting as well.

\begin{theorem}
	For the \maxavg\, and \maxmax\, objectives,
	the distortion of any \randrand\, mechanism is at least $3$, even when the metric space is a line.
	\label{th:randrand-maxX-lower}
\end{theorem}
\begin{proof}
	Consider any \randrand\ mechanism $\mech = (\fin,\fov)$. We construct an instance with candidates $\candidates = \{c_1, c_2, c_3\}$ and voters $\voters = \{v_1, v_2\}$, where each voter belongs to a distinct group: $v_1 \in g_1$ and $v_2 \in g_2$. The instance is constructed on the line metric as follows (refer to \Cref{fig:randrand-maxX}):

	\begin{itemize}
		\item Candidates $c_1$, $c_2$, and $c_3$ are located at positions $-1$, $0$, and $1$, respectively.
		\item Voters $v_1$ and $v_2$ are located at positions $-0.5$ and $0.5$, respectively.
		\item The preference profile of each voter is $\profile_1 = (c_1, c_2, c_3)$ for $v_1$ and $\profile_2 = (c_3, c_2, c_1)$ for $v_2$.
	\end{itemize}

	Trivially, the preference profiles are consistent with the distances in \Cref{fig:randrand-maxX}. According to \Cref{obs:singlevoter}, candidates $c_1$ and $c_3$ are chosen as the representatives of groups $g_1$ and $g_2$, respectively, and then mechanism $\mech$ must select one of them as the final winner.
	
	We have
	\(
	\cost(c_1) = \cost(c_3) = \frac{3}{2} \text{ and } \cost(c_2) = \frac{1}{2}.
	\)
	Clearly, $c_2$ is the optimal candidate.
	The distortion of $\mech$ is:
	\begin{align*}
	\distortion(\mech) &\ge \frac{\min\bigl(\cost(c_1), \cost(c_3)\bigr)}{\cost(\opt)} \\
	& = \frac{\frac{3}{2}}{\cost(c_2)} \\
	&= 3.
	\end{align*}
	\begin{figure}[t]
		\centering
		\input{figs/randrand-maxX.tex} 
		\caption{An example used in the proof of \Cref{th:randrand-maxX-lower}.}
		\label{fig:randrand-maxX}
	\end{figure}
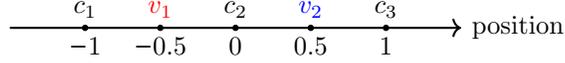
\end{proof}
\begin{theorem}
	For the \avgmax\, objective, the distortion of 	any $\randrand$ mechanism is at least $3 - \tfrac{2}{n}$ (equivalently, \( 3 - \tfrac{2}{m} \)), even when the metric space is a line.
	\label{th:randrand-avgmax-lower}
\end{theorem}
\begin{proof}
We construct an instance with a set of \( m \) candidates, \( \candidates = \{ c_1, c_2, \dots, c_m \} \), and a set of \( n = m \) voters, \( \voters = \{\voter_1, \voter_2, \dots, \voter_n\} \), all belonging to a single group. Each voter \( \voter_i \) has a preference profile $\profile_i$ that ranks candidates cyclically, starting with $c_i$ as their top choice:
	\[
	\profile_i = \left(c_i, c_{i+1}, \dots, c_m, c_1, c_2, \dots, c_{i-1} \right).
	\]
Now, we construct $ m $ instances, denoted \( \instance_1, \instance_2, \dots, \instance_m \), on the line metric. 
Across all instances, the voter set $\voters$, candidate set $\candidates$, and the preference profile $\profile$ are identical; they differ only in the arrangement of voters and candidates within the underlying metric space.

For any instance $\instance_{i}, \text{ where } 1 \leq i \leq m$, configuration of the line metric is as follows (see \Cref{fig:max_lower-maxX}):

\begin{itemize}
	\item Candidate \( c_i \) is located at position \( 0 \), while all other candidates are located at \( 1 \).
	\item Voter \( \voter_i \) is located at position \( -0.5 \), and all other voters are located at \( 0.5 \).
\end{itemize}
	It is straightforward to verify that each constructed instance is consistent with the specified preference profile. Since there is a single group, the \avgmax\, objective simplifies to \maxx. For each instance $\instance_{i}$, we have $\cost(c_i) = \frac{1}{2}, \text{ and } \cost(c_{j}) =\tfrac{3}{2}, $ for \(1 \le j \le m\) (where \(i \ne j\)). Clearly, $c_i$ is the optimal candidate in $\instance_{i}$.
	
	Now, consider any \randrand\ mechanism \(\mech \). Let $p_i$ denote the probability that $\mech$ selects $c_i$ as the winner, where $\sum_{i=1}^{m} p_i = 1$.
	For instance $\instance_{i}$, the mechanism's expected cost is
	\[\expected{\cost(\mech(\instance_{i}))} 
    =\sum_{j=1}^{m} p_j \cdot \cost(c_j)
	=\sum_{j \ne i} p_j \cdot \frac{3}{2} + p_i \cdot \frac{1}{2} = \frac{3}{2}(1-p_i)+ \frac{p_i}{2} = \frac{3}{2} - p_i. \]
	It follows that
	\begin{align*}
		\frac{\expected{\cost(\mech(\instance_{i}))}}{\cost(\opt)} = \frac{\frac{3}{2} - p_i}{\frac{1}{2}} 
		= 3-2p_i.
	\end{align*}
Since the total probability must sum up to 1, there exists some index $i$ such that $p_i \leq \frac{1}{m}=\frac{1}{n}.$ 
Finally, we obtain the following lower bound on the distortion of $\mech$
		\[\distortion(\mech) \geq 3-\frac{2}{m} = 3-\frac{2}{n}. \]
	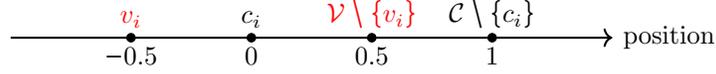
\begin{figure}[t]
		\centering
		\input{figs/max_lower-maxX.tex}
		\caption{Configuration of the candidates and voters in instance $\instance_i$ used in the proof of \Cref{th:randrand-avgmax-lower}.}
		\label{fig:max_lower-maxX}
	\end{figure}
\end{proof}



In the centralized setting, where all voters belong to a single group $(k=1)$, the \avgmax\, objective simplifies to the \maxx\, objective. Thus, the instance and analysis from the proof of \Cref{th:randrand-avgmax-lower}, which already considers the single-group case, apply directly.
As the number of voters $n$ increases, the distortion approaches $3$. Therefore, for any constant $\varepsilon > 0$, an instance can be constructed with a sufficiently large number of voters ($n > \tfrac{2}{\varepsilon}$) such that the distortion of any randomized voting rule exceeds $3 - \varepsilon$. This yields the following corollary.

\begin{corollary}[of \Cref{th:randrand-avgmax-lower}]
		For the \maxx\, objective in the centralized setting, the distortion of any randomized voting rule is at least $3 - \varepsilon$, for any constant $\varepsilon > 0$, even when the metric is a line.
		\label{cor:max-Xmax-lower}
\end{corollary}
\begin{theorem}
	For general metric spaces and the \avgavg\, objective, the distortion of any \randrand\ mechanism is at least $3 - \tfrac{2}{n}$ (equivalently $3 - \tfrac{2}{k}$).
	\label{th:randrand-avgavg-lower}
\end{theorem}
\begin{proof}
    Consider a \randrand\ mechanism $\mech$. We construct an instance with a set of $m = k+1$ candidates, $\candidates = \{c_1, c_2, \ldots, c_{m = k+1}\}$, a set of $n=k$ voters $\voters = \{v_1, v_2, \ldots, v_{n = k}\}$, and $k$ single-voter groups, $g_i = \{v_i\}$ for $1 \le i \le k$. Each voter \( \voter_i \) has a preference profile $\profile_i$, with $c_i$ as the top choice, immediately followed by $c_m$, denoted as
    \(
    \profile_i = \movetofirst{\movetofirst{\sigma}{c_m}}{c_i},
    \)
    where $ \sigma $ is an ordering of the candidates.

    We now construct a connected graph $G$ with $n+m$ vertices, denoted $u_1, u_2, \ldots, u_{n+m}$, where the shortest-path distances in $G$ define the underlying metric space $\textsf{d}$. Each voter and candidate is placed on one of the vertices. The construction of $G$ is as follows (see \Cref{fig:randrand-avg-graph} for an illustration):
    
    \begin{itemize}
    	\item Place candidate $c_m$ at vertex $u_1$. 
    	
    	\item For each $1 \le i \le k$, place voter $v_i$ at vertex $u_{2i}$, and candidate $c_i$ at vertex $u_{2i+1}$.
    	
    	\item For each $1 \le i \le k$, add an edge between  $u_1$ and $u_{2i}$, and another edge between $u_{2i}$ and $u_{2i+1}$. This forms $k$ branches extending from the central vertex $u_1$.
    \end{itemize}
	See \Cref{tab:rand-rand-avg} for the corresponding distances. For all $1 \le i \le k$, we have 
	\( \topp{v_i} = c_i \).   
    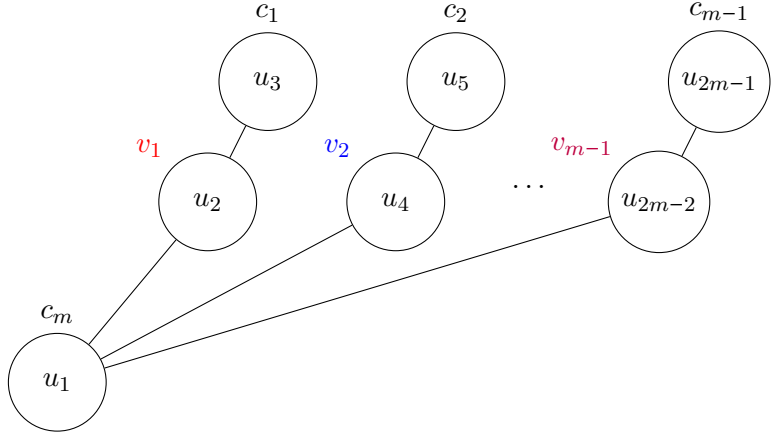
\begin{figure}[t]
    	\centering
    	\input{figs/randrand-avg-graph.tex}
    	\caption{Tree graph used in the proof of \Cref{th:randrand-avgavg-lower}. Different voter groups are distinguished by distinct colors.}
    	\label{fig:randrand-avg-graph}
    \end{figure}
    \begin{table}[t]
    	\centering
    	\begin{tabular}{|>{\columncolor{gray!20}}c|*{3}{>{\centering\arraybackslash}p{0.75cm}}|}
    		\hline
    		\rowcolor{gray!20}
    		$\dis{\cdot}{\cdot}$ & \bm{$c_i$} & \bm{$c_j$} & \bm{$c_m$} \\
    		\hline
    		\bm{$v_i$} & 1 & 3 & 1 \\
    		\hline
    	\end{tabular}
    	\vspace{0.25cm}
    	\caption{For any \(1 \leq i, j \leq k\) with \(i \neq j\), the shortest-path distances between candidates and voters, derived from tree graph $G$ in the proof of \Cref{th:randrand-avgavg-lower}.}
    	\label{tab:rand-rand-avg}
    \end{table}
	Therefore, the representative of group $g_i$ is candidate $c_i$ (by \Cref{obs:singlevoter}). Consequently $c_m$ is not the representative of any group, and thus cannot be the winner of mechanism $\mech$. According to the definition of the \avgavg\ objective, we have
    $\cost(c_i) = \frac{3(n-1) + 1}{n} = 3 - \frac{2}{n}$ for $1 \leq i < m$, and
	$\cost(c_m) = 1$.
	Clearly, $c_m$ is the optimal candidate. We obtain the lower bound on the distortion of $\mech$ as follows:
	\begin{align*}
		 \distortion(\mech) &\ge \min_{1 \le i < m} \bigg({\frac{\cost(c_i)}{\cost(\opt)}} \bigg) \\
		 & = \frac{3 - \frac{2}{n}}{\cost(c_m)} \\
		 & = 3 - \frac{2}{n}.
    \end{align*}
\end{proof}

%% file: figs/randrand-maxX.tex
\begin{tikzpicture}[x=2cm, y=1cm, font=\small]
	
	\draw[->, thick] (-1.5, 0) -- (1.5, 0) node[right] {position};
	
	\coordinate (V1) at (-0.5,0);
	\coordinate (V2) at (0.5,0);
	\coordinate (C2) at (0,0);
	\coordinate (C1) at (-1,0);
	\coordinate (C3) at (1,0);
	
	\filldraw[black] (V1) circle (1pt) node[anchor=south] {\textcolor{red}{$v_1$}};
	\filldraw[black] (V2) circle (1pt) node[anchor=south] {\textcolor{blue}{$v_2$}};
	
	\filldraw[black] (C1) circle (1pt) node[anchor=south] {\textcolor{black}{$c_1$}};
	\filldraw[black] (C2) circle (1pt) node[anchor=south] {\textcolor{black}{$c_2$}};
	\filldraw[black] (C3) circle (1pt) node[anchor=south] {\textcolor{black}{$c_3$}};
	\foreach \x/\lab in {-1/$-1$, -0.5/{$-0.5$}, 0/0, 0.5/{$0.5$}, 1/{$1$}} {
		\node[below] at (\x, 0) {\lab};
	}
\end{tikzpicture}

%% file: figs/max_lower-maxX.tex
\begin{tikzpicture}[x=2cm, y=1cm, font=\small, scale=1.6]
	
	\draw[->, thick] (-1, 0) -- (1.5, 0) node[right] {position};

	\coordinate (V1) at (-0.5,0);
	\coordinate (C1) at (0,0);
	\coordinate (OTV) at (0.5,0);
	\coordinate (OTC) at (1,0);
	
	\filldraw[black] (V1) circle (1pt) node[anchor=south] {\textcolor{red}{$v_i$}};
	\filldraw[black] (C1) circle (1pt) node[anchor=south] {\textcolor{black}{$c_i$}};
	\filldraw[black] (OTV) circle (1pt) node[anchor=south] {\textcolor{red}{$\voters \setminus \{\voter_i\}$}};
	\filldraw[black] (OTC) circle (1pt) node[anchor=south] {\textcolor{black}{$\candidates \setminus \{\candidate_i\}$}};
	
	\foreach \x/\lab in {-0.5/{$-0.5$}, 0/0, 0.5/{$0.5$}, 1/{$1$}} {

		\node[below] at (\x, 0) {\lab};
	}
	
\end{tikzpicture}

%% file: figs/randrand-avg-graph.tex
\begin{tikzpicture}[ every node/.style={}, 
	label distance=0pt
	]
	
	\node[minimum size=1.3cm, circle, draw] (u1) at (0,0) {$u_1$};
	\node[above=0pt of u1, draw=none] {$c_m$};
	
	\node[minimum size=1.3cm, circle, draw] (u2) at (2,2.4) {$u_2$};
	\node[above left=0pt of u2, draw=none] {\textcolor{red}{$v_1$}};
	
	\node[minimum size=1.3cm, circle, draw] (u3) at (2.8,4.0) {$u_3$};
	\node[above=0pt of u3, draw=none] {$c_1$};
	\draw (u1) -- (u2) -- (u3);
	
	\node[minimum size=1.3cm, circle, draw] (u4) at (4.5,2.4) {$u_4$};
	\node[above left=0pt of u4, draw=none] {\textcolor{blue}{$v_2$}};
	
	\node[minimum size=1.3cm, circle, draw] (u5) at (5.3,4.0) {$u_5$};
	\node[above=0pt of u5, draw=none] {$c_2$};
	\draw (u1) -- (u4) -- (u5);
	
	\node[draw=none] at (6.3, 2.6) {$\dots$};
	
	\node[minimum size=1.3cm, circle, draw] (uN1) at (8,2.4) {$u_{2m-2}$};
	\node[above left=0pt of uN1, draw=none] {\textcolor{purple}{$v_{m-1}$}};
	
	\node[minimum size=1.3cm, circle, draw] (uN2) at (8.8,4.0) {$u_{2m-1}$};
	\node[above=0pt of uN2, draw=none] {$c_{m-1}$};
	\draw (u1) -- (uN1) -- (uN2);

\end{tikzpicture}

%% file: src/detdet.tex

In this section, we study the distortion of deterministic distributed mechanisms, providing both lower and upper bounds in general metric spaces. \citet{anshelevich2022distortion} proposed the $\alpha$-in-$\beta$-over mechanism ($\malphabeta$), which allows any candidate---not just representatives---to be selected as the final winner. We adopt their approach here.


%% file: src/detdet_upper.tex
We establish improved upper bounds for deterministic mechanisms with respect to the \avgmax\, and \maxmax\, objectives. Let $\fbeta$ be a deterministic voting rule with distortion at most $\beta$, and let $\fun$ be a deterministic voting rule that satisfies the property of pareto efficiency. We improve the best known upper bound for the \avgmax\, objective from 11, as proved by \citet{anshelevich2022distortion}, to 7 by applying mechanism $(\fun, \fbeta)$. For the \maxmax\, objective, we improve the previous upper bound of 5 to 3 by the $\mad$ mechanism. We begin with the simplest case, \maxmax, and move towards the slightly more intricate \avgmax.

\begin{theorem}
	For the \maxmax\, objective in general metric spaces, we have $\distortion(\mad) \leq 3$.
	\label{th:detdet_maxmax}
\end{theorem}
\begin{proof}
	Consider an instance $\instance = (\voters, \candidates, \groups, \profile, \textsf{d})$ and the \emph{Arbitrary Dictator} mechanism ($\mech = \mad$), which selects the top-ranked candidate of an arbitrary voter $\voter$ as the final winner; that is, $\w=\topp{\voter}$. We follow a strategy roughly analogous to that in \Cref{th:randrand_maxmax}.
	\begin{align*}
		\cost(\topp{\voter})
		&= \dis{\vstarstar{\topp{\voter}}}{\topp{\voter}}& \text{(Definition of $\maxmax$)}\\
		&\leq \dis{\vstarstar{\topp{\voter}}}{\voter} + \dis{\voter}{\topp{\voter}} & \text{(Triangle Inequality)}\\
		&\leq \dis{\vstarstar{\topp{\voter}}}{\voter} + \dis{\voter}{\opt} & \text{(\Cref{prop:topi})}\\
		&\leq \dis{\vstarstar{\topp{\voter}}}{\opt} + \dis{\opt}{\voter} + \dis{\voter}{\opt} & \text{(Triangle Inequality)}\\
		& = \dis{\vstarstar{\topp{\voter}}}{\opt} + 2\dis{\voter}{\opt} \\
		& \le \dis{\vstarstar{\topp{\voter}}}{\opt} + 2\dis{\vstarstar{\opt}}{\opt} & \text{(\Cref{prop:maxmax})} \\
		&\leq 3 \dis{\vstarstar{\opt}}{\opt} & \text{(\Cref{prop:maxmax})} \\
		&= 3\cost(\opt) & \text{(Definition of $\maxmax$)}.
	\end{align*}
\end{proof}

\begin{theorem}
	 For the \avgmax\, objective in general metric spaces, we have $\distortion((\fun, \fbeta)) \leq 2\beta + 3$.
	\label{th:det_avgmax}
\end{theorem}
\begin{proof}
Consider an instance $\instance = (\voters, \candidates, \groups, \profile, \textsf{d})$ and a \detdet\  mechanism $\mech = (\fun, \fbeta)$.   
By the property of pareto efficiency, for each group $g$, there exists a voter $\voter_g \in g$ who prefers the representative $\w_g$ over the optimal candidate $\opt$. We have
\begin{align*}
	\cost(\w) 
	&= \frac{1}{k} \sum_{g \in \groups} \dis{\vstar{\w}{g}}{\w} & \quad\text{(Definition of \avgmax)} \\
	&\leq \frac{1}{k} \sum_{g \in \groups} \dis{\vstar{\w}{g}}{\opt} + \dis{\opt}{\w} & \quad\text{(Triangle Inequality)}\\
	&\leq \frac{1}{k} \sum_{g \in \groups} \dis{\vstar{\opt}{g}}{\opt} + \dis{\opt}{\w} & \quad\text{(\Cref{prop:Xmax})}\\
	&= \cost(\opt) + \frac{1}{k} \sum_{g \in \groups} \dis{\opt}{\w} & \quad\text{(Definition of \avgmax)}\\
	&\leq \cost(\opt) + \frac{1}{k} \sum_{g \in \groups} \dis{\opt}{\w_g} + \dis{\w}{\w_g} & \quad\text{(Triangle Inequality)} \\
	&\leq \cost(\opt) + (\beta + 1) \cdot \frac{1}{k} \sum_{g \in \groups} \dis{\opt}{\w_g} & \quad\text{(\Cref{prop:fbeta})}\\
	&\leq \cost(\opt) + (\beta + 1) \cdot \frac{1}{k} \sum_{g \in \groups} \dis{\opt}{v_g} + \dis{v_g}{\w_g} & \quad\text{(Triangle Inequality)}\\
	&\leq \cost(\opt) + 2(\beta + 1) \cdot \frac{1}{k} \sum_{g \in \groups} \dis{\opt}{v_g} & \quad\text{($\dis{v_g}{\w_g} \leq \dis{v_g}{\opt}$)}\\
	&\leq \cost(\opt) + 2(\beta + 1) \cdot \frac{1}{k} \sum_{g \in \groups} \dis{\opt}{\vstar{\opt}{g}} & \quad\text{(\Cref{prop:Xmax})}\\
	&= (2\beta + 3) \cost(\opt) & \quad\text{(Definition of \avgmax).}
\end{align*}
\end{proof}

We can apply $\fpmpar$ as both the in-group and over-group voting rules. Assuming each voter is at a distance of 0 from her top choice, this yields a distortion of $\beta = 2$, as shown in \citep{gkatzelis2020resolving}.
Combined with \Cref{th:det_avgmax}, we conclude that $\mech = (\fpmpar,\fpmpar)$ is a \detdet\ mechanism that satisfies the bound in \Cref{cor:detdet_avgmax}.
\begin{corollary}[of \Cref{th:det_avgmax}]
	For the \avgmax\, objective in general metric spaces, there exists a \detdet\ mechanism with distortion at most 7.
	\label{cor:detdet_avgmax}	
\end{corollary}

%% file: src/detdet_lower.tex

In this section, we present a lower bound of 5 on the distortion of any \detdet\ mechanisms under the \maxavg\ objective, improving upon the previous bound of $2+\sqrt{5}$, which achieved by \citet{anshelevich2022distortion}. Following the framework in that paper, our analysis applies the over-group voting rule across the set of candidates ($\candidates$), rather than the group representatives ($R$).

\begin{theorem}
	For general metric spaces and the \maxavg\, objective, the distortion of any \detdet\ mechanism is at least $5$.
	\label{th:detdet-maxavg-lower}
\end{theorem}


\begin{proof}
	Consider a \detdet\, mechanism $\mech=(\fin, \fov)$. We construct an instance with a set of candidates $\candidates = \{c_1, c_2, c_3, c_4\}$, a set of voters $\voters = \{v_1, v_2, v_3, v_4\}$, and $2$ groups $g_1 = \{v_1, v_2\}$ and $g_2 = \{v_3, v_4\}$.
	Let \( \sigma \) be an arbitrary ordering of the candidates. Without loss of generality, assume that \( c_1 \) has an in-degree of at least \( \left\lceil \frac{m - 1}{2} \right\rceil = 2 \) in the tournament \( \tour{\fin}{\candidates}{\sigma} \), such a candidate is guaranteed to exist by \Cref{pr:bias-tournamnet-indegree}. Further, let \( c_2 \) and \( c_3 \) be the two candidates with directed edges toward \( c_1 \) in the tournament.

	\begin{figure}[t]
		\centering
		\scalebox{1}{\input{figs/det-maxavg-lower-graph.tex}}
		\caption{An illustration of the graph $G$ used in the proof of \Cref{th:detdet-maxavg-lower}. Different voter groups are distinguished by distinct colors.}
		\label{fig:det-maxavg-lower-graph}
	\end{figure}
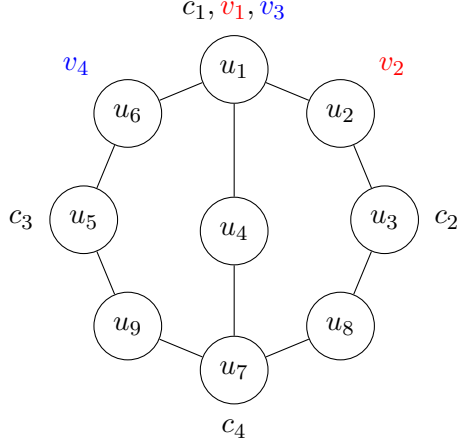
	
	We now construct a connected graph $G$ with $9$ vertices, denoted $u_1, u_2, \ldots u_9$, where the shortest-path distances in $G$ define the underlying metric space $\textsf{d}$. Each voter and candidate is placed at one of the vertices (a single vertex may host multiple entities).
	The graph configuration is as shown in \Cref{fig:det-maxavg-lower-graph}. 
		
		
	
	
	Pairwise distances between the candidates and voters are presented in \Cref{tab:detdet-maxavg-table}. Now, we define the preference profiles in \Cref{tab:detdet-maxavg-profiles}, which are consistent with both the shortest-path distances in graph $G$ and the input of tournament $\tour{\fin}{\candidates}{\sigma}$. Note that multiple profiles may be consistent.
	
	\begin{table}[t]
		\centering
		\begin{tabular}{|>{\columncolor{gray!20}}c|*{4}{>{\centering\arraybackslash}p{1cm}}|}
			\hline
			\rowcolor{gray!20}
			$\dis{\cdot}{\cdot}$ & \bm{$c_1$} & \bm{$c_2$} & \bm{$c_3$} & \bm{$c_4$} \\
			\hline
			\bm{$v_1$} & 0 & 2 & 2 & 2 \\
			\bm{$v_2$} & 1 & 1 & 3 & 3 \\
			\bm{$v_3$} & 0 & 2 & 2 & 2 \\
			\bm{$v_4$} & 1 & 3 & 1 & 3 \\
			\hline
		\end{tabular}
		\vspace{0.3cm}
		\caption{The shortest-path distances between candidates and voters, as derived from the graph $G$ in the proof of \cref{th:detdet-maxavg-lower}.}
		\label{tab:detdet-maxavg-table}
	\end{table}
	
		\begin{table}[t]
		\centering
		\begin{tabular}{|>{\columncolor{gray!20}}c|*{1}{>{\centering\arraybackslash}p{3.5cm}}|}
			\hline
			\rowcolor{gray!20}
			Voter & Preference Profile \\
			\hline
			\bm{$v_1$} & $\movetofirst{\movetofirst{\sigma}{c_2}}{c_1}$ \\
			\bm{$v_2$} & $\movetofirst{\movetofirst{\sigma}{c_1}}{c_2}$ \\
			\bm{$v_3$} & $\movetofirst{\movetofirst{\sigma}{c_3}}{c_1}$ \\
			\bm{$v_4$} & $\movetofirst{\movetofirst{\sigma}{c_1}}{c_3}$ \\
			\hline
		\end{tabular}
		\vspace{0.3cm}
		\caption{The voter preference profiles used in the proof of \cref{th:detdet-maxavg-lower}.}
		\label{tab:detdet-maxavg-profiles}
	\end{table}
	
	Since candidates $c_2$ and $c_3$ defeat $c_1$ in $\tour{\fin}{\candidates}{\sigma}$, they must serve as the representative of groups $g_1$ and $g_2$, respectively. Thus, the set of representatives is $R=\{c_2, c_3\}$.
	
	At the second stage of the distributed voting process, we consider two instances:
	\begin{itemize}
		\item \(\instance_1\): The preference profiles of \(c_2\) and \(c_3\) are 
		\((c_2, c_1, c_4, c_3)\) and \((c_3, c_1, c_4, c_2)\), respectively.
		\item \(\instance_2\): The preference profiles of \(c_2\) and \(c_3\) are 
		\((c_2, c_4, c_1, c_3)\) and \((c_3, c_4, c_1, c_2)\), respectively.
	\end{itemize}
	
	It is straightforward to verify that both \(\instance_1\) and \(\instance_2\) are consistent with the metric space defined earlier. 
	In both cases, the preference profiles of \(c_2\) and \(c_3\) follow the pattern
	\[
	(c_2, c_a, c_b, c_3) 
	\quad \text{and} \quad 
	(c_3, c_a, c_b, c_2),
	\]
	where \((c_a, c_b)\) is some ordering of \((c_1, c_4)\). 
	
	If the rule selects the first-ranked or fourth-ranked candidate, then \(c_1\) is not chosen.  
	If the rule selects the second-ranked candidate (\(c_a\)), then in \(\instance_2\) we again exclude \(c_1\).  
	If the rule selects the third-ranked candidate (\(c_b\)), then in \(\instance_1\) we exclude \(c_1\).  
	Therefore, in every case, there exists an instance in which the mechanism selects the winner from the set \(\{c_2, c_3, c_4\}\).

	By the definition of the \maxavg\ objective, we have
	\(\cost(c_1) = \tfrac{1}{2}\)
	and 
	\(\cost(c_2) = \cost(c_3) = \cost(c_4) = \tfrac{5}{2}.\)
	Thus, \(c_1\) is the optimal candidate. 
	The distortion of mechanism \(\mech\) is obtained as follows:
	\begin{align*}
		\distortion(\mech)
		& \ge \min \Bigl( \frac{ \cost(c_2), \cost(c_3), \cost(c_4) }{\cost(\opt)} \Bigr) \\
		& \ge \frac{\frac{5}{2}}{\cost(c_1)} \\
		& = 5.
	\end{align*}
\end{proof}

%% file: figs/det-maxavg-lower-graph.tex
\begin{tikzpicture}[scale=1, every node/.style={}, label distance=0pt]
	\foreach \name/\angle in {u_1/90, u_2/45, u_3/0, u_8/-45, u_7/-90, u_9/-135, u_5/180, u_6/135}
	\node[circle, draw, minimum size=0.9cm] (\name) at (\angle:2) {$\name$};
	
	\node[circle, draw, minimum size=0.9cm] (u_4) at (0,-0.15) {$u_4$};
	
	\node[above=2pt of u_1] {$c_1, \textcolor{red}{v_1}, \textcolor{blue}{v_3}$};
	\node[above right=2pt of u_2] {\textcolor{red}{$v_2$}};
	\node[right=2pt of u_3] {$c_2$};
	\node[below=2pt of u_7] {$c_4$};
	\node[left=2pt of u_5] {$c_3$};
	\node[above left=2pt of u_6] {\textcolor{blue}{$v_4$}};
	
	\draw (u_1) -- (u_2) -- (u_3) -- (u_8) -- (u_7) -- (u_9) -- (u_5) -- (u_6) -- (u_1); 
	\draw (u_1) -- (u_4) -- (u_7); 
\end{tikzpicture}

%% file: src/euclidean.tex

In this section, we focus on the Euclidean metric and establish lower bounds on the $\avgavg$ distortion for both the \randrand\ and \randdet\ mechanisms, as presented in  \Cref{th:euc-avgavg-randrand,th:euc-avgavg-randdet}, respectively.

\begin{theorem}
	For the \avgavg\, objective in Euclidean space, the distortion of any \randrand\ mechanism is at least $\sqrt{5} - \varepsilon$, for every constant $\varepsilon > 0$.
\label{th:euc-avgavg-randrand}
\end{theorem}


\begin{proof}
	Consider any \randrand\ mechanism $\mech$.
	Let $\unusedvar$ be a positive integer. Consider an instance in $(\unusedvar + 1)-$dimensional Euclidean space, ${\mathbb{R}}^{\unusedvar+1}$, with $\unusedvar + 2$ candidates,  denoted $c_1, c_2, \dots, c_{\unusedvar+2}$, and $k = \unusedvar + 1$ groups, each with a single voter $v_i$ for $1 \le i \le \unusedvar + 1$.
	We construct the instance as follows:
	\begin{itemize}
		\item Let $q_i$ be the point in $\mathbb{R}^{\unusedvar+1}$ whose $i$-th coordinate is 1 and all other coordinates are 0, for $1 \leq i \leq \unusedvar+1$.
		
		\item Place candidate $c_i$ at point $q_i$ for each $1 \leq i \leq \unusedvar+1$.
		
		
		\item The final candidate, 
		$c_{\unusedvar+2}$, is placed at the centroid of the other candidates; $\left(\frac{1}{\unusedvar+1}, \frac{1}{\unusedvar+1}, \dots, \frac{1}{\unusedvar+1}\right)$.
		
		\item In the $i$-th group,  the single voter $v_i$ is positioned at the midpoint between their corresponding candidate, $c_i$ and the centroid candidate, $c_{\unusedvar+2}$. Indeed, 
		each voter $v_i$ is located at a point where the $i$-th coordinate is $\tfrac{\unusedvar+2}{2(\unusedvar+1)}$ and all other coordinates are $\tfrac{1}{2(\unusedvar+1)}$. 
		Note that each voter’s preference profile is structured so that
		the top-ranked candidate of voter $v_i$ is $c_i$, consistent with the underlying Euclidean space.
	\end{itemize} 
In particular, in 3D space ($\unusedvar = 2$), the instance lies within an equilateral triangle with vertices $(1,0,0)$, $(0,1,0)$, and $(0,0,1)$, as illustrated in \Cref{fig:euc-randrand}.
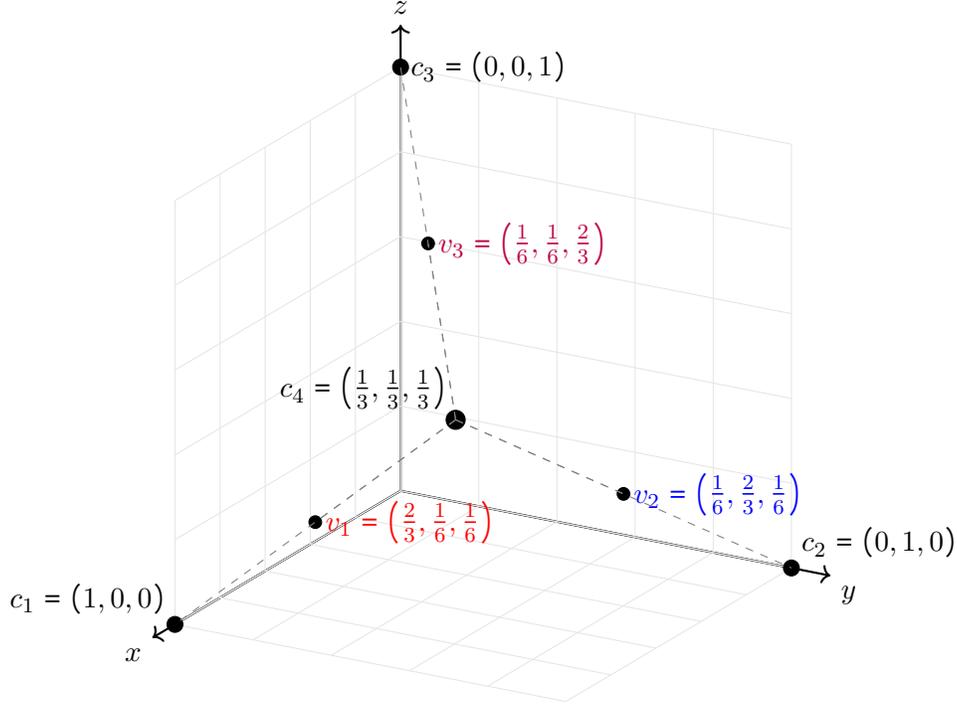
\begin{figure}[t]
	\centering
	\input{figs/euc-randrand.tex}
	\caption{A $3-dimensional$ Euclidean model ($\unusedvar = 2$) illustrating the geometric lower bound construction for the \randrand\ mechanisms under the \avgavg\, objective. Candidates $c_1$, $c_2$, and $c_3$ are located at the unit basis vectors, with  $c_4 = \left(\tfrac{1}{3}, \tfrac{1}{3}, \tfrac{1}{3}\right)$ at the centroid of the triangle they form. Voters $v_1$, $v_2$, and $v_3$ are positioned at the midpoints between the centroid $c_4$ and their top-ranked candidate $c_i$. Different group voters are indicated via distinct colors.}
	\label{fig:euc-randrand}
\end{figure}

For all $1 \leq i \leq \unusedvar + 1$, we have
\begin{align*}
	\dis{c_i}{v_i} &= \dis{c_{\unusedvar+2}}{v_i} \\
	&= \sqrt{\left(\frac{1}{2(\unusedvar + 1)}\right)^2 \unusedvar + \left(\frac{\unusedvar}{2(\unusedvar + 1)}\right)^2} \\
	&= \sqrt{\frac{\unusedvar}{4(\unusedvar + 1)}}.
\end{align*} 
Moreover, for all $1 \leq i,j \leq \unusedvar + 1$ (where $i \neq j$), we have
\begin{align*}
	\dis{c_i}{v_j} &= \sqrt{\left(\frac{1}{2(\unusedvar + 1)}\right)^2 (\unusedvar - 1) + \left(\frac{\unusedvar + 2}{2(\unusedvar + 1)}\right)^2 + \left(\frac{2\unusedvar + 1}{2\unusedvar + 2}\right)^2} \\
	&= \sqrt{\frac{5\unusedvar + 4}{4\unusedvar + 4}}.
\end{align*}
By the definition of the \avgavg\, objective, the cost of each candidate is the average distance to all $\unusedvar+1$ voters. Thus, we conclude that
\begin{align*}
	&\cost(c_{\unusedvar+2}) = \frac{1}{2} \sqrt{\frac{\unusedvar}{\unusedvar+1}}, \\
	&\cost(c_i) = \frac{\unusedvar\sqrt{ \frac{5\unusedvar+4}{4\unusedvar+4}} + \frac{1}{2} \sqrt{\frac{\unusedvar}{\unusedvar+1}}}{\unusedvar+1} \quad (1 \leq i \leq \unusedvar+1).
\end{align*}

Clearly, the optimal candidate is $c_{\unusedvar+2}$. According to \Cref{obs:singlevoter}, the representative of the $i$-th group is candidate $c_i$. Finally, the mechanism selects the winner from among the candidates $c_1, c_2, \dots, c_{\unusedvar+1}$. A lower bound on the distortion of the mechanism $\mech$ is obtained as follows $(1 \leq i \leq \unusedvar+1)$:
\begin{align*}
	\distortion\left( \mech \right) &\geq \frac{\cost(c_i)}{\cost(c_{\unusedvar+2})} \\
	& = \frac{\frac{\unusedvar}{\unusedvar+1}\sqrt{ \frac{5\unusedvar+4}{4\unusedvar+4}} + \frac{1}{2(\unusedvar+1)} \sqrt{\frac{\unusedvar}{\unusedvar+1}}}{\frac{1}{2} \sqrt{\frac{\unusedvar}{\unusedvar+1}}} .
\end{align*}
As $\unusedvar \to \infty$, the ratio approaches $\sqrt{5} \approx 2.236$.
Therefore, for any $\varepsilon > 0$, we can construct an instance with distortion greater than $\sqrt{5} - \varepsilon$.
\end{proof}

\begin{theorem}
	For the \avgavg\, objective in Euclidean space, the distortion of any \randdet\ mechanism is at least $2+\sqrt{5}-\varepsilon$, for every constant $\varepsilon > 0$.
 \label{th:euc-avgavg-randdet}
\end{theorem}


\begin{proof}
Consider a \randdet\ mechanism $\mech = (\fin, \fov)$, a set of $2m$ candidates $\candidates = \{ c_1, c_2, \ldots, c_{2m} \} $, and an arbitrary ordering $\sigma$ over them. 
By \Cref{pr:bias-tournamnet-indegree}, the tournament \( \tour{\fin}{\candidates}{\sigma} \), must have a candidate with in-degree at least \( \left\lceil \frac{2m - 1}{2} \right\rceil = m \). Without loss of generality, let  \( c_{m+1} \) be such a candidate and let \( c_1, c_2, \ldots, c_m \) be \( m \) candidates that have directed edges toward \( c_{m+1} \) in the tournament.
We now construct the following instance with $k = m$ groups in $(m + 1)-$dimensional Euclidean space:
\begin{itemize}
	\item Let $q_i$ be the point in $\mathbb{R}^{m+1}$ whose $i$-th coordinate is 1 and all other coordinates are 0, for $1 \leq i \leq m+1$.
	
	\item Place candidate $c_i$ at point $q_i$ for each $1 \leq i \leq m$ and candidate $c_{m+1}$ at the centroid $\left( \frac{1}{m+1}, \frac{1}{m+1}, \dots, \frac{1}{m+1} \right)$.
	
	\item Place candidate $c_i$ at point $q_{m+1}$ for each $m+2 \leq i \leq 2m$.
	\item In the $i$-th group ($1 \leq i \leq m$), there are two voters: 
	\begin{enumerate}
		\item Voter $v_{2i-1}$ is located at the centroid, 
which is the same position as candidate $c_{m+1}$.
		\item Voter $v_{2i}$ is located exactly at the midpoint between candidates $c_{m+1}$ and $c_i$, with coordinates equal to $\tfrac{m+2}{2(m+1)}$ in the $i$-th dimension and $\tfrac{1}{2(m+1)}$ in all other dimensions.
	\end{enumerate}
\item The ordinal preferences of the $v_{2i-1}$ and $v_{2i}$ are defined as
$\pi_{2i-1} = \movetofirst{\sigma}{\movetofirst{c_{i}}{c_{m+1}}}$
and 
$\pi_{2i} = \movetofirst{\sigma}{\movetofirst{c_{m+1}}{c_{i}}}$.
These preferences are consistent with the underlying Euclidean space:
\begin{enumerate}
	\item The distance from \( v_{2i-1} \) to all candidates except \( c_{m+2} \) is equal.
	\item The distance from \( v_{2i-1} \) to \( c_{m+2} \) is zero.
	\item Voter \( v_{2i} \) is closer to candidates \( c_{m+1} \) and \( c_i \) than to any other candidates, and is equidistant from all remaining ones.
	\item The distance from \( v_{2i} \) to candidates \( c_{m+1} \) and \( c_i \) is equal.
\end{enumerate}
\end{itemize} 
When $m = 2$ the instance lies within an equilateral triangle with vertices at $(1,0,0)$, $(0,1,0)$, and $(0,0,1)$, as shown in \Cref{fig:euc-randdet}.
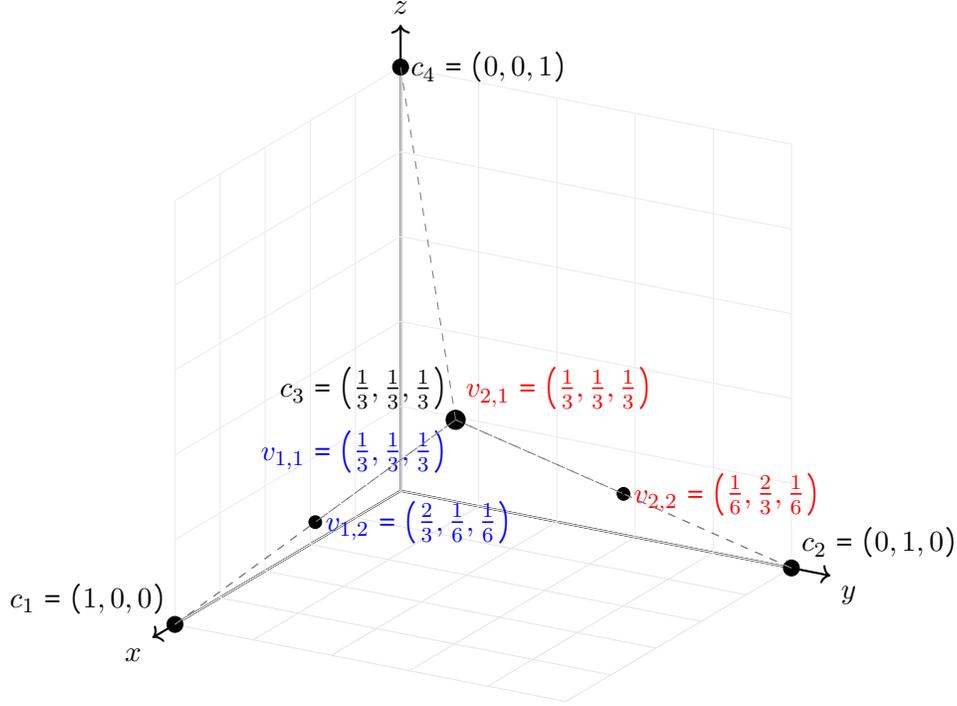
\begin{figure}[t]
	\centering
	\input{figs/euc-randdet.tex}
	\caption{An illustration of the constructed instance when $m = 2$. The candidates are positioned at the corners and centroid of the 3D simplex (i.e., the equilateral triangle embedded in $\mathbb{R}^3$). Candidate $c_3$ is placed at the centroid, representing the candidate with high in-degree in $\tour{\fin}{\candidates}{\sigma}$. Each group contains two voters: $v_{2i-1}$ is located at the centroid, while $v_{2i}$ is placed at the midpoint between $c_3$ and $c_i$ for $i = 1,2$. Different group voters are indicated via distinct colors.}
	\label{fig:euc-randdet}
\end{figure}
For all $1 \leq i,j \leq m$, we have
\begin{align*}
	\dis{c_i}{c_{m+1}} &= \dis{c_i}{v_{2j-1}} 
	\\&= \sqrt{\left( \frac{1}{m+1} \right)^2 m + \left( \frac{m}{m+1} \right)^2} 
	\\&= \sqrt{\frac{m}{m+1}},
\end{align*}	
for all $1 \leq i,j \leq m$ (where $i \neq j$), we have
\begin{align*}
	\dis{c_i}{v_{2j}} &= \sqrt{\left( \frac{1}{2(m+1)} \right)^2 (m-1) + \left( \frac{m+2}{2(m+1)} \right)^2 + \left( \frac{2m+1}{2m+2} \right)^2} \\
	&= \sqrt{\frac{5m+4}{4m+4}},
\end{align*}	
and for all $1 \leq i \leq m$, we have
\begin{align*}	
	\dis{c_i}{v_{2i}} &= \sqrt{\left( \frac{1}{2(m+1)} \right)^2 m + \left( \frac{m}{2(m+1)} \right)^2} \\&= \sqrt{\frac{m}{4(m+1)}}.
\end{align*}
By the definition of the \avgavg\, objective, we conclude that
\begin{align*}
	&\cost(c_{m+1}) = \tfrac{m\left(\tfrac{0 + \sqrt{\tfrac{m}{4(m+1)}}}{2}\right)}{m} = \frac{1}{4} \sqrt{\frac{m}{m+1}}, \\
	&\cost(c_i) = \frac{ \tfrac{ \sqrt{\frac{m}{m+1}} + \sqrt{\frac{5m+4}{4m+4}}}{2} (m-1) + \tfrac{ \sqrt{\frac{m}{m+1}} + \sqrt{\frac{m}{4(m+1)}} }{2}}{m} &(1 \leq i \leq m).
\end{align*}
Clearly, the optimal candidate is $c_{m+1}$. By the definition of $\tour{\fin}{\candidates}{\sigma}$, the representative of group $i$ ($1 \leq i \leq m$) is $c_i$. Therefore, the mechanism selects the final winner from among the candidates $c_1, c_2, \dots, c_m$. 
A lower bound on the distortion of the mechanism $\mech$ is obtained as follows ($1 \leq i \leq m$):
\begin{align*}
	\distortion\left(\mech\right) &\geq \frac{\cost(c_i)}{\cost(c_{m+1})} \\
	&= 2 + \frac{m-1}{m} \sqrt{\frac{5m+4}{m}} + \frac{1}{m} &(1 \leq i \leq m),
\end{align*}
As $m \to \infty$, the ratio approaches $2+\sqrt{5} \approx 4.236$. Therefore, for any $\varepsilon > 0$, we can construct an instance with distortion greater than $2 + \sqrt{5} - \varepsilon$.
\end{proof}

%% file: figs/euc-randrand.tex
\tdplotsetmaincoords{70}{120}
\begin{tikzpicture}[tdplot_main_coords, scale=6]
	
	\draw[->, thick] (0,0,0) -- (1.1,0,0) node[anchor=north east]{$x$};
	\draw[->, thick] (0,0,0) -- (0,1.1,0) node[anchor=north west]{$y$};
	\draw[->, thick] (0,0,0) -- (0,0,1.1) node[anchor=south]{$z$};
	
	\colorlet{gridcolor}{gray!20}
	
	\foreach \x in {0,0.2,...,1}
	\draw[gridcolor] (\x,0,0) -- (\x,1,0);
	\foreach \y in {0,0.2,...,1}
	\draw[gridcolor] (0,\y,0) -- (1,\y,0);
	
	\foreach \y in {0,0.2,...,1}
	\draw[gridcolor] (0,\y,0) -- (0,\y,1);
	\foreach \z in {0,0.2,...,1}
	\draw[gridcolor] (0,0,\z) -- (0,1,\z);
	
	\foreach \x in {0,0.2,...,1}
	\draw[gridcolor] (\x,0,0) -- (\x,0,1);
	\foreach \z in {0,0.2,...,1}
	\draw[gridcolor] (0,0,\z) -- (1,0,\z);
	
	\coordinate (C1) at (1,0,0);
	\coordinate (C2) at (0,1,0);
	\coordinate (C3) at (0,0,1);
	\coordinate (C4) at (1/3,1/3,1/3);  
	
	\filldraw[black]  (C1) circle (0.5pt) node[anchor=south east] {\textcolor{black}{$c_1 = (1,0,0)$}};
	\filldraw[black]  (C2) circle (0.5pt) node[anchor=south west] {\textcolor{black}{$c_2 = (0,1,0)$}};
	\filldraw[black]  (C3) circle (0.5pt) node[anchor=west] {\textcolor{black}{$c_3 = (0,0,1)$}};
	\filldraw[black] (C4) circle (0.6pt) node[anchor=south east] {\textcolor{black}{$c_4 = \left(\tfrac{1}{3},\tfrac{1}{3},\tfrac{1}{3}\right)$}};
	
	\coordinate (V1) at (2/3,1/6,1/6);
	\coordinate (V2) at (1/6,2/3,1/6);
	\coordinate (V3) at (1/6,1/6,2/3);
	
	\filldraw[black] (V1) circle (0.4pt) node[anchor=west] {\textcolor{red}{$v_1 = \left(\tfrac{2}{3},\tfrac{1}{6},\tfrac{1}{6}\right)$}};
	\filldraw[black] (V2) circle (0.4pt) node[anchor=west] {\textcolor{blue}{$v_2 = \left(\tfrac{1}{6},\tfrac{2}{3},\tfrac{1}{6}\right)$}};
	\filldraw[black] (V3) circle (0.4pt) node[anchor=west] {\textcolor{purple}{$v_3 = \left(\tfrac{1}{6},\tfrac{1}{6},\tfrac{2}{3}\right)$}};
	
	\foreach \P in {C1,C2,C3,V1,V2,V3}
	\draw[dashed, gray] (C4) -- (\P);
	
\end{tikzpicture}

%% file: figs/euc-randdet.tex
\tdplotsetmaincoords{70}{120}

\begin{tikzpicture}[tdplot_main_coords, scale=6]
	
	\coordinate (C1) at (1,0,0);
	\coordinate (C2) at (0,1,0);
	\coordinate (C3) at (1/3,1/3,1/3);
	\coordinate (C4) at (0,0,1);
	\coordinate (V11) at (1/3,1/3,1/3); 
	\coordinate (V21) at (1/3,1/3,1/3); 
	\coordinate (V12) at ({(1+1/3)/2}, {(0+1/3)/2}, {(0+1/3)/2}); 
	\coordinate (V22) at ({(0+1/3)/2}, {(1+1/3)/2}, {(0+1/3)/2}); 
	
	\draw[->, thick] (0,0,0) -- (1.1,0,0) node[anchor=north east]{$x$};
	\draw[->, thick] (0,0,0) -- (0,1.1,0) node[anchor=north west]{$y$};
	\draw[->, thick] (0,0,0) -- (0,0,1.1) node[anchor=south]{$z$};
	
	\foreach \x in {0,0.2,...,1} \draw[gray!20] (\x,0,0) -- (\x,1,0);
	\foreach \y in {0,0.2,...,1} \draw[gray!20] (0,\y,0) -- (1,\y,0);
	\foreach \y in {0,0.2,...,1} \draw[gray!15] (0,\y,0) -- (0,\y,1);
	\foreach \z in {0,0.2,...,1} \draw[gray!15] (0,0,\z) -- (0,1,\z);
	\foreach \x in {0,0.2,...,1} \draw[gray!15] (\x,0,0) -- (\x,0,1);
	\foreach \z in {0,0.2,...,1} \draw[gray!15] (0,0,\z) -- (1,0,\z);

	\filldraw[black]   (C1) circle (0.5pt) node[anchor=south east] {\textcolor{black}{$c_1 = (1,0,0)$}};
	\filldraw[black]   (C2) circle (0.5pt) node[anchor=south west] {\textcolor{black}{$c_2 = (0,1,0)$}};
	\filldraw[black]  (C3) circle (0.6pt) node[anchor=south east] {\textcolor{black}{$c_3 = \left(\tfrac{1}{3},\tfrac{1}{3},\tfrac{1}{3}\right)$}};
	\filldraw[black](C4) circle (0.5pt) node[anchor=west] {\textcolor{black}{$c_4 = (0,0,1)$}};
	
	\filldraw[black]  (V11) circle (0.4pt) node[anchor=north east] {\textcolor{blue}{$v_{1,1} = \left(\tfrac{1}{3},\tfrac{1}{3},\tfrac{1}{3}\right)$}};
	\filldraw[black]  (V21) circle (0.4pt) node[anchor=south west] {\textcolor{red}{$v_{2,1} = \left(\tfrac{1}{3},\tfrac{1}{3},\tfrac{1}{3}\right)$}};
	\filldraw[black](V12) circle (0.4pt) node[anchor=west] {\textcolor{blue}{$v_{1,2} = \left(\tfrac{2}{3},\tfrac{1}{6},\tfrac{1}{6}\right)$}};
	\filldraw[black](V22) circle (0.4pt) node[anchor=west] {\textcolor{red}{$v_{2,2} = \left(\tfrac{1}{6},\tfrac{2}{3},\tfrac{1}{6}\right)$}};
	
	\draw[dashed, gray] (C1) -- (C3);
	\draw[dashed, gray] (C2) -- (C3);
	\draw[dashed, gray] (C4) -- (C3);
	\draw[dashed, gray] (C1) -- (V12);
	\draw[dashed, gray] (C3) -- (V12);
	\draw[dashed, gray] (C2) -- (V22);
	\draw[dashed, gray] (C3) -- (V22);
	
\end{tikzpicture}

%% file: src/discussion.tex

In this paper, we have initiated the study of metric distortion in single-winner distributed voting under randomized mechanisms (\randrand\ and \randdet) for many different objectives. We also have improved upon previous results for deterministic mechanisms (\detdet). 

Although our work presents an almost complete picture in the distortion of distributed voting problem, it reveals several promising directions for future research. A significant challenge about our work leaves open lies in analyzing the \detrand\ mechanisms, where random decisions in the first stage are followed by deterministic ones in the second. Our primary tool, the Bias Tournament technique, is incompatible with the randomized first stage of \detrand.
Currently, our understanding is  confined to basic results inherited from \randrand\ (for the lower bounds) and \detdet\ (for the upper bounds). Therefore, developing an analytical approach to precisely resolve the distortion of the \detrand\ mechanisms represents a compelling direction for future work.

Within the scope of our work, another possible direction could be to close the remaining narrow gaps between the lower and upper bounds presented in \Cref{tab:ourresults}, particularly for the \avgavg\ and \avgmax\ objectives in \detdet, as well as the \avgavg\ and \avgmax\ objective in \randrand.
Another direction is to investigate more structured spaces, such as line metrics and Euclidean spaces. While many of our results also hold on the line (i.e., one-dimensional Euclidean space), some of the bounds we obtain are not tight—or even close to tight—when considered in the context of line metrics or Euclidean spaces. Specifically, we can study the avg-avg cost function for rand-rand, rand-det, and det-det, and analyze them in Euclidean and line metrics. Additionally, we can study the max-avg cost function in Euclidean and line metrics for det-det. The bounds in these structured spaces may differ from those in general metric spaces.
Another natural extension is to investigate distributed mechanisms in a cardinal setting, where agents have access to exact distances, instead of solely the ordinal rankings induced by those distances.

Going beyond the single-winner setting, one could study the distortion of distributed mechanisms that select committees comprising a specified number of alternatives. 
Another intriguing direction for future research is to investigate how agents' strategic behavior impacts distributed distortion. The goal could be to understand whether it is possible to design distributed mechanisms that are both strategyproof and capable of achieving low distortion.